\documentclass[oribibl]{llncs}

\usepackage{booktabs} 

\usepackage{amsfonts}
\usepackage{enumerate,paralist}
\usepackage{mathtools}
\usepackage{amsmath}
\usepackage{stmaryrd}

\usepackage{tikz,pgffor}
\usetikzlibrary{arrows}
\usetikzlibrary{shapes}
\usetikzlibrary{calc}
\usetikzlibrary{automata}
\usetikzlibrary{positioning}

\tikzstyle{location} = [
    rectangle,
    rounded corners,
    draw=black,
    very thick,
    minimum height=2em,
    inner sep=0pt,
    text centered
]
\tikzstyle{tran}  = [
    draw,
    ->,
    >=stealth,
    rounded corners
]
\tikzstyle{nchoice}  = [
    draw,
    >=stealth,
    rounded corners
]
\tikzstyle{pchoice}  = [
    draw,
    ->,
    >=stealth,
    rounded corners,
    dashed
]
\tikzstyle{dec}   = [inner sep=0pt]
\tikzstyle{mode}  = [
    shape=circle,
    draw,
    inner sep=0pt,
    minimum size=5mm
]

\newcommand{\Rset}{\mathbb{R}}
\newcommand{\Nset}{\mathbb{N}}
\newcommand{\Zset}{\mathbb{Z}}
\newcommand{\ap}{\mbox{\sl AP}}
\newcommand{\opt}{\mbox{\sl opt}}

\newcommand{\clocks}{\mathcal{X}}
\newcommand{\val}[1]{\mbox{\sl Val}\left(#1\right)}
\newcommand{\reset}[2]{{#1}{\left[#2:=0\right]}}
\newcommand{\add}[2]{{#1}{+}{#2}}
\newcommand{\zero}{\mathbf{0}}
\newcommand{\true}{\mathbf{true}}
\newcommand{\false}{\mathbf{false}}
\newcommand{\supp}[1]{{\mbox{\sl supp}}(#1)}
\newcommand{\dist}[1]{{\mathcal{D}}{\left(#1\right)}}
\newcommand{\sat}[1]{{\llbracket}{#1}{\rrbracket}}

\newcommand{\pta}{\mathcal{C}}
\newcommand{\locs}{L}
\newcommand{\loc}{\ell}
\newcommand{\acts}{\mbox{\sl Act}}
\newcommand{\inv}{\mbox{\sl inv}}
\newcommand{\enab}{\mbox{\sl enab}}

\newcommand{\prob}{\mbox{\sl prob}}
\newcommand{\lbfunc}{\mathcal{L}}
\newcommand{\clcons}[1]{\mbox{\sl CC}\left(#1\right)}

\newcommand{\states}{S}
\newcommand{\trans}{\rightarrow}
\newcommand{\tran}[3]{{#1}{\xrightarrow{#2}}{#3}}
\newcommand{\probk}{\mathbf{P}}

\newcommand{\fnpaths }[1]{\mbox{\sl Paths}^*_{#1}}
\newcommand{\infpaths}[1]{\mbox{\sl Paths}^\omega_{#1}}

\newcommand{\initloc }[1]{\mbox{\sl init}\left(#1\right)}
\newcommand{\lastloc }[1]{\mbox{\sl last}\left(#1\right)}

\newcommand{\fnpath}{\rho}
\newcommand{\infpath}{\pi}

\newcommand{\length}[1]{\left| #1 \right|}

\newcommand{\probm}{\mathbb{P}}

\newcommand{\cyl}{\mbox{\sl Cyl}}

\newcommand{\dta}{\mathcal{A}}
\newcommand{\dtloc}{q}
\newcommand{\dtclocks}{\mathcal{Y}}
\newcommand{\cstates}{Q}
\newcommand{\alphabet}{\Sigma}
\newcommand{\rules}{\Delta}

\newcommand{\trfunc}{\kappa}
\newcommand{\dtatr}[3]{{#1}{\xRightarrow{#2}}{#3}}

\newcommand{\run}[3]{{#1}_{#2}\left(#3\right)}
\newcommand{\iconfig}{\left(\dtloc^*,\mathbf{0}\right)}

\newcommand{\product}[2]{{#1}{\otimes}{#2}}
\newcommand{\pr}[2]{\mathfrak{p}_{#1}^{#2}}
\newcommand{\pfunc}{\mathcal{T}}
\newcommand{\sfunc}{\theta}

\newcommand{\exttuples}{\mathcal{T}}

\newcommand{\infset}[1]{\mbox{\sl inf}  ( #1 )}
\newcommand{\trace }[1]{\mbox{\sl trace}( #1 )}
\newcommand{\traj  }[1]{\mbox{\sl traj} ( #1 )}
\newcommand{\rabin}{\mathcal{F}}

\newcommand{\fstates}{F}

\newcommand{\Lang}[2] {
    \mbox{\sl AccPaths}
        _{#1}
        ^{#2}
}

\newcommand{\rabinp}[1]{\mbox{\sl RPaths}_{#1}}

\newcommand{\LangCsAqF}{
    \Lang
        {\pta,\sigma}
        {\dta,\dtloc}
}

\newcommand{\TLang}{
    \pfunc\left(
        \LangCsAqF
    \right)
}

\newcommand{\TAcc}{
    \rabinp{\sfunc\left(\sigma\right)}
}
\newcommand{\accept}[2]{
    \mbox{\bf ACC} \left(
      {#1},
      {#2}
    \right)
}
\newcommand{\acccept}[3]{
    \mbox{\bf ACC} \left(
      {#1},
      {#2},
      {#3}
    \right)
}

\newcommand{\Region}[1]{
  \mbox{Reg}[
    {#1}
  ]
}
\newcommand{\WAIT}[0]{\mbox{\sl WAIT}}
\newcommand{\WORK}[1]{\mbox{\sl WORK}_{#1}}
\newcommand{\DONE}[1]{\mbox{\sl DONE}_{#1}}
\newcommand{\request}[1]{\mbox{\sl Req}_{#1}}
\newcommand{\p}[1]{\mbox{\sl p}_{#1}}

\def\INIT{  \mbox{\sl INIT} }
\def\FAIL{  \mbox{\sl FAIL} }
\newcommand{\q}[1]{\mbox{\sl q}_{#1}}

\newcommand{\location}[5]{
    \node[location] (#1) at #2 {
        \begin{tabular}{c}
            \ensuremath{#3} \\
            \hline
            \ensuremath{#4} \\
            \ensuremath{#5}
        \end{tabular}
    }
}

\newcommand{\PairV}[2]{#1#2}
\newcommand{\PairS}[2]{({#1},{#2})}

\newcounter{row}
\newcounter{col}

\newcommand\setrow[3]{
  \setcounter{col}{1}
  \foreach \n in {#1, #2, #3} {
    \edef\x{\value{col} - 0.5}
    \edef\y{3.5 - \value{row}}
    \node[anchor=center] at (\x, \y) {\n};
    \stepcounter{col}
  }
  \stepcounter{row}
}

\newcommand{\nta}{\mathcal{A}}
\newcommand{\qinit}{q_{\mbox{\sl\scriptsize init} }}
\newcommand{\qstart}{q_{\mbox{\sl\scriptsize start}}}
\newcommand{\ntaap}[1]{b_{#1}}

\newcommand{\tra}{\mathcal{A}}

\begin{document}
\title{Verifying Probabilistic Timed Automata Against \\Omega-Regular Dense-Time Properties}

\author{
    Hongfei Fu\inst{1} 
    \and 
    Yi Li\inst{2} 
    \and 
    Jianlin Li\inst{3,4}
}
\institute{
    Shanghai Jiao Tong University, Shanghai, China
    \and
    Department of Informatics, School of Mathematical Sciences, Peking University, Beijing, China 
    \and
    State Key Laboratory of Computer Science, Institute of Software, Chinese Academy of Sciences, Beijing, China
    \and
    College of Computer Science and Technology, Nanjing University of Aeronautics and Astronautics, Nanjing, China
}

\maketitle

\vspace{-2em}

\begin{abstract}
Probabilistic timed automata (PTAs) are timed automata (TAs) extended with discrete probability distributions.
They serve as a mathematical model for a wide range of applications that involve both stochastic and timed behaviours.
In this work, we consider the problem of model-checking linear \emph{dense-time} properties over {PTAs}.
In particular, we study linear dense-time properties that can be encoded by TAs with infinite acceptance criterion.
First, we show that the problem of model-checking PTAs against deterministic-TA specifications can be solved through a product construction.
Based on the product construction, we prove that the computational complexity of the problem with deterministic-TA specifications is EXPTIME-complete.
Then we show that when relaxed to general (nondeterministic) TAs, the model-checking problem becomes undecidable.
Our results substantially extend state of the art with both the dense-time feature and the nondeterminism in TAs.
\end{abstract}

\vspace{-1em}

\vspace{-2em}
\section{Introduction}
\vspace{-1em}
Stochastic timed systems are systems that exhibit both timed and stochastic behaviours.
Such systems play a dominant role in many applications~\cite{DBLP:books/daglib/0020348}, hence
addressing fundamental issues such as safety and performance over these systems are important.
\emph{Probabilistic timed automata} (PTAs)~\cite{DBLP:journals/fmsd/NormanPS13,DBLP:journals/tcs/Beauquier03,DBLP:journals/tcs/KwiatkowskaNSS02} serve as a good mathematical model for these systems.
They extend the well-known model of timed automata~\cite{DBLP:journals/tcs/AlurD94} (for nonprobabilistic timed systems) with discrete probability distributions, and Markov Decision Processes (MDPs)~\cite{PutermanMDP} (for untimed probabilistic systems) with timing constraints.

Formal verification of PTAs has received much attention in recent years~\cite{DBLP:journals/fmsd/NormanPS13}.
For branching-time model-checking of PTAs, the problem is reduced to computation of reachability probabilities over MDPs through well-known finite abstraction for timed automata (namely \emph{regions} and \emph{zones})~\cite{JensenPTA,DBLP:journals/tcs/Beauquier03,DBLP:journals/tcs/KwiatkowskaNSS02}.
Advanced techniques for branching-time model checking of PTAs such as inverse method and symbolic method have been further explored in  ~\cite{DBLP:journals/fmsd/AndreFS13,DBLP:journals/iandc/KwiatkowskaNSW07,DBLP:conf/formats/KwiatkowskaNP09,DBLP:conf/formats/JovanovicKN15}.
Extension with \emph{cost} or \emph{reward}, resulting in \emph{priced} PTAs, has also been well investigated.
Jurdzinski~\emph{et al.}~\cite{DBLP:conf/concur/JurdzinskiKNT09} and Kwiatkowska~\emph{et al.}~\cite{DBLP:journals/fmsd/KwiatkowskaNPS06} proved that several notions of accumulated or discounted cost are computable over priced PTAs, while
cost-bounded reachability probability over priced PTAs is shown to be undecidable by Berendsen \emph{et al.}~\cite{DBLP:conf/tamc/BerendsenCJ09}.
Most verification algorithms for PTAs have been implemented in the model checker PRISM~\cite{DBLP:conf/cav/KwiatkowskaNP11}. Computational complexity of several verification problems for PTAs has been studied, for example, \cite{DBLP:journals/ipl/LaroussinieS07,DBLP:journals/lmcs/JurdzinskiSL08,DBLP:conf/concur/JurdzinskiKNT09}.

For linear-time model-checking, much less is known.
As far as we know, the only relevant result is by Sproston~\cite{DBLP:conf/qest/Sproston11} who proved that the problem of model-checking PTAs against linear \emph{discrete-time} properties encoded by \emph{untimed} deterministic omega-regular automata (e.g., Rabin automata) can be solved by a product construction.
In his paper, Sproston first devised a production construction that produces a PTA out of the input PTA and the automaton;
then he proved that the problem can be reduced to omega-regular verification of MDPs through maximal end components.

In this paper, we study the problem of model-checking linear \emph{dense-time} properties over PTAs.
Compared with discrete-time properties, dense-time properties take into account timing constraints, and therefore is more expressive and applicable to time-critical systems.
Simultaneously, verification of dense-time properties is more challenging since it requires to involve timing constraints.
The extra feature of timing constraints also brings more theoretical difficulty, e.g.,
timed automata~\cite{DBLP:journals/tcs/AlurD94} (TAs) are generally not determinizable,
which is in contrast to untimed automata (such as Rabin or Muller automata).

We focus on linear dense-time properties that can be encoded by TAs.
Due to the ability to model dense-time behaviours,
TAs can be used to model real-time systems, while they can also act as language recognizers for timed omega-regular languages.
Here we treat TAs as language recognizers for timed paths from a PTA, and study
the problem of computing the minimum or maximum probability that a timed path from the PTA is accepted by the TA.
The intuition is that a TA can recognize the set of ``good'' (or ``bad'') timed paths emitting from a PTA,
so the problem is to compute the probability that the PTA behaves in a good (or bad) manner.
The relationship between TAs and linear temporal logic (e.g., Metric Temporal Logic~\cite{DBLP:journals/rts/Koymans90}) is studied in~\cite{DBLP:conf/lics/OuaknineW05,DBLP:journals/jacm/AlurFH96}.

\smallskip
\noindent{\em Our Contributions.} We distinguish between the subclass of \emph{deterministic} TAs (DTAs) and general \emph{nondeterministic} TAs.
DTAs are the deterministic subclass of TAs.
Although the class of DTAs is weaker than general timed automata, it can recognize a wide class of formal timed languages, and express interesting linear dense-time properties which cannot be expressed in branching-time logics (cf.~\cite{DBLP:journals/tse/DonatelliHS09}).
We consider Rabin acceptance condition as the infinite acceptance criterion for TAs.
We first show that the problem of model-checking PTAs against DTA specifications with Rabin acceptance condition
can be solved through a nontrivial product construction which tackles the integrated feature of timing constraints and randomness. From the product construction, we further prove that the problem is EXPTIME-complete.
Then we show that the problem becomes undecidable when one considers general TAs.
Our results substantially extend previous ones (e.g.~\cite{DBLP:conf/qest/Sproston11}) with both the dense-time feature and the nondeterminism in TAs.

Due to lack of space, detailed proofs of several results are put in the appendix.

\vspace{-1.5em}
\section{Preliminaries}
\vspace{-1em}

We denote by $\Nset$, $\Nset_0$, $\Zset$, and $\Rset$ the sets of all positive
integers, non-negative integers, integers and real numbers, respectively.
For any infinite word $w=b_0b_1\dots$ over an alphabet $\Sigma$, we denote by $\infset{w}$ the set of symbols 
in $\Sigma$ that occur infinitely often in $w$.

%
%
A \emph{clock} is a variable for a nonnegative real number. Below we fix a finite set $\clocks$ of clocks.

\smallskip \noindent{\em Clock Valuations.}
A \emph{clock valuation} is a function $\nu:\clocks\rightarrow [0,\infty)$. The set of clock valuations
is denoted by $\val{\clocks}$.
Given a clock valuation $\nu$, a subset $X\subseteq\clocks$ of clocks and a non-negative real number $t$, we let (i) $\reset{\nu}{X}$ be the clock valuation such that $\reset{\nu}{X}(x)=0$ for $x\in X$ and $\reset{\nu}{X}(x)=\nu(x)$ otherwise, and (ii) $\add{\nu}{t}$ be the clock valuation such that $(\add{\nu}{t})(x)=\nu(x)+t$ for all $x\in\clocks$.
We denote by $\zero$ the clock valuation such that $\zero(x)=0$ for $x\in\clocks$.

\smallskip \noindent{\em Clock Constraints.} The set $\clcons{\clocks}$ of \emph{clock constraints}  over $\clocks$ is generated by the following grammar:
$
\phi:=~\true~\mid~x\le d~\mid ~c\le x~\mid~x+c\le y+d~\mid ~\neg\phi~\mid~\phi\wedge\phi
$
where $x,y\in\clocks$ and $c,d\in\Nset_0$.
We write $\false$ for a short hand of $\neg\true$.
The satisfaction relation $\models$ between valuations $\nu$ and clock constraints $\phi$ is defined through substituting every $x\in\clocks$ appearing in $\phi$ by $\nu(x)$ and standard semantics for logical connectives.
For a given clock constraint $\phi$, we denote by $\sat{\phi}$ the set of all clock valuations that satisfy $\phi$.

\vspace{-1.5em}
\subsection{Probabilistic Timed Automata}
\vspace{-0.6em}

A \emph{discrete probability distribution} over a countable non-empty set $U$ is a function $q:U\rightarrow[0,1]$ such that $\sum_{z\in U}q(z)=1$.
The \emph{support} of $q$ is defined as $\supp{q}:=\{z\in U\mid q(z)>0\}$.
We denote the set of discrete probability distributions over $U$ by $\dist{U}$.

\vspace{-0.6em}
\begin{definition}[Probabilistic Timed Automata~\cite{DBLP:journals/fmsd/NormanPS13}]
A \emph{probabilistic timed automaton} (PTA) $\pta$ is a tuple
\begin{equation}\label{eq:pta}
\pta=\left(\locs, \loc^*, \clocks, \acts, \inv, \enab,  \prob, \ap, \lbfunc\right)
\end{equation}
where :
\begin{compactitem}
\item $\locs$ is a finite set of \emph{locations};
\item $\loc^*\in\locs$ is the \emph{initial} location;
\item $\clocks$ is a finite set of \emph{clocks};
\item $\acts$ is a finite set of \emph{actions};
\item $\inv:\locs\rightarrow\clcons{\clocks}$ is
an \emph{invariant condition};
\item $\enab:\locs\times\acts\rightarrow\clcons{\clocks}$ is an \emph{enabling condition};
\item $\prob:\locs\times\acts\rightarrow\dist{2^{\clocks}\times\locs}$ is a \emph{probabilistic transition function};
\item $\ap$ is a finite set of \emph{atomic propositions};
\item $\lbfunc:\locs\rightarrow 2^{\ap}$ is a \emph{labelling function}.
\end{compactitem}
\end{definition}
%
W.l.o.g, we consider that both $\acts$ and $\ap$ is disjoint from $[0,\infty)$. Below we fix a PTA $\pta$.
The semantics of PTAs is as follows.

\smallskip \noindent{\em States and Transition Relation.}
A \emph{state} of $\pta$ is a pair $(\loc, \nu)$ in $\locs\times\val{\clocks}$ such that $\nu\models \inv(\loc)$.
The set of all states is denoted by $\states_\pta$.
The \emph{transition relation} $\trans$ consists of all triples $((\loc,\nu),a,(\loc',\nu'))$ satisfying
the following conditions:
\begin{compactitem}
\item $(\loc,\nu), (\loc',\nu')$ are states and $a\in\acts\cup [0,\infty)$;
\item if $a\in [0,\infty)$ then $\nu+\tau\models \inv(\loc)$ for all $\tau\in [0, a]$ and $(\loc',\nu')=(\loc,\nu+a)$;
\item if $a\in\acts$ then $\nu\models\enab(\loc,a)$ and there exists a pair $(X, \loc'')\in\supp{\prob(\loc,a)}$ such that $(\loc',\nu')=(\loc'',\reset{\nu}{X})$.
\end{compactitem}
By convention, we write $\tran{s}{a}{s'}$ instead of $(s,a,s')\in\trans$.
We omit `$\pta$' in `$\states_\pta$' if the underlying context is clear.

\smallskip \noindent{\em Probability Transition Kernel.}
The \emph{probability transition kernel} $\probk$ is the function $\probk:\states\times\acts\times\states\rightarrow[0,1]$ such that
\vspace{-0.3cm}
\begin{align*}
    & \probk((\loc,\nu),a,(\loc',\nu'))
        = 
        & \begin{cases}
            1 &
                \mbox{if }\tran{(\loc,\nu)}{a}{(\loc',\nu')}\mbox{ and } a\in [0,\infty)\\
            \sum_{Y\in B}\prob(\loc,a)(Y,\loc') &
                \mbox{if }\tran{(\loc,\nu)}{a}{(\loc',\nu')}\mbox{ and } a\in\acts \\
            0 &
                \mbox{otherwise}
        \end{cases}
\end{align*}
where $B:=\{X\subseteq\clocks\mid \nu'=\reset{\nu}{X}\}$.

\smallskip \noindent{\em Well-formedness.} We say that $\pta$ is \emph{well-formed} if for every state $(\loc,\nu)$ and action $a\in\acts$ such that $\nu\models\enab(\loc,a)$ and every $(X,\loc')\in \supp{\prob(\loc,a)}$, one has that $\reset{\nu}{X}\models\inv(\loc')$.
The well-formedness is to ensure that when an action is enabled, the next state after taking this action will always be legal. In the following, we always assume that the underlying PTA is well-formed. Non-well-formed PTAs can be repaired into well-formed PTAs~\cite{DBLP:journals/iandc/KwiatkowskaNSW07}.

\smallskip \noindent{\em Paths.}
A \emph{finite path} $\fnpath$ (under $\pta$) is a finite sequence
$
\left\langle s_0,a_0,s_1,\dots,a_{n-1},s_n\right\rangle~~(n\ge 0)
$
in
$\states\times{\left((\acts\cup[0,\infty))\times \states\right)}^*$
such that (i) $s_0=(\loc^*,\zero)$,
(ii) $a_{2k}\in [0,\infty)$ (resp. $a_{2k+1}\in \acts$) for all integers $0\le k\le \frac{n}{2}$ (resp. $0\le k\le \frac{n-1}{2}$) and
(iii) for all $0\le k\le n-1$, $\tran{s_k}{a_k}{s_{k+1}}$.
The length $\length{\fnpath}$ of $\fnpath$ is defined by $\length{\fnpath}:=n$.
An \emph{infinite path} (under $\pta$) is an infinite sequence
$
\left\langle s_0,a_0,s_1,a_1,\dots\right\rangle
$
in
${\left(\states\times(\acts\cup[0,\infty))\right)}^\omega$
such that for all $n\in\Nset_0$, the prefix $\left\langle s_0,a_0,\dots,a_{n-1},s_n\right\rangle$ is a finite path.
The set of finite (resp. infinite) paths  under $\pta$ is denoted by $\fnpaths{\pta}$ (resp. $\infpaths{\pta}$).

\smallskip \noindent{\em Schedulers.}
A \emph{(deteterministic) scheduler} is a function $\sigma$ from the set of finite paths into $\acts\cup [0,\infty)$ such that for all finite paths $\fnpath=s_0a_0\dots s_n$,
(i) $\sigma(\fnpath)\in\acts$ (resp. $\sigma(\fnpath)\in  [0,\infty)$) if $n$ is odd (resp. even) and (ii)
there exists a state $s'$ such that $\tran{s_n}{\sigma(\fnpath)}{s'}$.

\smallskip\noindent{\em Paths under Schedulers.}
A finite path $s_0a_0\dots s_n$ \emph{follows} a scheduler $\sigma$ if for all $0\le m< n$, $a_m=\sigma\left(s_0a_0\dots s_m\right)$.
An infinite path $s_0a_0s_1a_1\dots$ \emph{follows} $\sigma$ if for all $n\in\Nset_0$, $a_n=\sigma\left(s_0a_0\dots s_n\right)$.
The set of finite (resp. infinite) paths following a scheduler $\sigma$ is denoted by $\fnpaths{\pta,\sigma}$ (resp. $\infpaths{\pta,\sigma}$).
We note that the set $\fnpaths{\pta,\sigma}$ is countably infinite from definition.

\smallskip \noindent{\em Probability Spaces under Schedulers.}
Let $\sigma$ be any scheduler.
The probability space w.r.t $\sigma$ is defined as
$
(\Omega^{\pta,\sigma}, \mathcal{F}^{\pta,\sigma}, \probm^{\pta,\sigma})
$
where (i) $\Omega^{\pta,\sigma}:=\infpaths{\pta,\sigma}$, (ii) $\mathcal{F}^{\pta,\sigma}$ is the smallest sigma-algebra generated by all cylinder sets induced by finite paths for which
a finite path $\fnpath$ induces the cylinder set $\cyl(\fnpath)$ of all infinite paths in $\infpaths{\pta,\sigma}$ with $\fnpath$ being their (common) prefix,
and (iii) $\probm^{\pta,\sigma}$ is the unique probability measure such that for all finite paths $\fnpath=s_0a_0\dots a_{n-1}s_n$ in $\fnpaths{\pta,\sigma}$,
\vspace{-0.8em}
\[
\textstyle{\probm^{\pta,\sigma}(\cyl(\fnpath))=\prod_{k=0}^{n-1} \probk(s_k, \sigma(s_0a_0\dots a_{k-1}s_k), s_{k+1}).}
\]
\vspace{-0.8em}
For details see~\cite{DBLP:journals/tcs/KwiatkowskaNSS02}.

\smallskip \noindent{\em Zenoness and Time-Divergent Schedulers.}
An infinite path $\infpath=s_0a_0s_1a_1\dots$ is \emph{zeno} if $\sum_{n=0}^\infty d_n<\infty$, where $d_n:=a_n$ if $a_n\in [0,\infty)$ and $d_n:=0$ otherwise.
Then a scheduler $\sigma$ is \emph{time divergent} if $\probm^{\pta,\sigma}(\{\pi\mid\pi\mbox{ is zeno}\})=0$.
In the following, we only consider time-divergent schedulers.
The purpose is to eliminate non-realistic zeno behaviours (i.e., performing infinitely many actions within a finite amount of time).
\vspace{-3em}
\begin{figure}[]
    \centering
    \begin{minipage}[t]{0.5\linewidth}  
        \centering  
        \resizebox{1\textwidth}{!}{
            \def \bY {1}
\def \aY {4}

\def \tX {2}
\def \workX {4}
\def \decX  {6}
\def \doneX {8}

\begin{tikzpicture}[x = 2 cm,y = 0.8 cm]



\location{a1}{(\workX,\aY)}
    { \WORK{\alpha} }
    { x\le 10 }
    { \{ \alpha \} };

\node[dec] (adec) at (\decX,\aY) {$\bullet$};



\location{b1}{(\workX,\bY)}
    { \WORK{\beta} }
    { x\le 15 }
    { \{ \beta \} };

\node[dec] (bdec) at (\tX,\bY) {$\bullet$};





\draw[nchoice] (a1) to node[auto,below] {
    $\tau_\alpha,\true$
} (adec);

\draw[pchoice,bend right] (adec) to node[auto,above] {
    $ \{ x \}, 0.1$
} (a1);

\draw[pchoice] (adec) -- (\decX,\bY) to node[auto,above] {
    $ \{ x \},0.9$
} (b1);



\draw[nchoice] (b1) to node[auto,above] {
    $\tau_\beta,\true$
} (bdec);

\draw[pchoice,bend right] (bdec) to node[auto,below] {
    $ \{ x \}, 0.2$
} (b1);

\draw[pchoice] (bdec) -- (\tX,\aY) to node[auto] {
    $ \{ x \},0.8$
}  (a1);


\end{tikzpicture}
        }
        \caption{A Simple Task-Processing Example}
        \label{fig:pta}
    \end{minipage}
    \begin{minipage}[t]{0.45\linewidth}  
        \centering
        \resizebox{1\textwidth}{!}{
            \def \hY {8}
\def \aY {6}
\def \maY{5}
\def \mY {4}
\def \mbY{3}
\def \bY {2}

\def \lX {3}
\def \qX {4}
\def \rX {5}
\def \fX {8}

\begin{tikzpicture} [x = 0.8cm]

\node[mode,initial]  (i) at (0,  \mY  ) {$ \INIT $};
\node[mode,accepting](qa) at (\qX,\aY) {$ \q{\alpha} $};
\node[mode,accepting](qb) at (\qX,\bY) {$ \q{\beta } $};

\node[mode]          (f) at (\fX,\mY) {$ \FAIL $};

\draw[tran] (i) --(0,\aY) to node[auto,sloped,above] {
    $\{ \alpha \},\true, \{ y \}$
} (qa);
\draw[tran] (i) --(0,\bY) to node[auto,sloped,below] {
    $\{ \beta  \},\true, \{ y \}$
} (qb);

\draw[tran,bend left] (qa) to node[auto,sloped,below] {
    $\{ \beta \}, W_\beta < y, \emptyset$
} (f);
\draw[tran,bend right] (qb) to node[auto,sloped,above] {
    $\{ \alpha  \}, W_\alpha < y, \emptyset$
} (f);

\draw[tran] (qa) to node[auto,sloped,above] {
    $\{ \alpha \}, C_\alpha < y, \emptyset$
} (\fX,\aY) -- (f);
\draw[tran] (qb) to node[auto,sloped,below] {
    $\{ \beta  \}, C_\beta < y, \emptyset$
} (\fX,\bY) -- (f);


\draw[tran] (qa) -- (\rX,\hY) to node[auto,above] {
    $\{ \alpha \}, y \le C_\alpha, \emptyset$
} (\lX,\hY) -- (qa);

\draw[tran] (qb) -- (\rX,0) to node[auto,below] {
    $\{ \beta \}, y \le C_\beta, \emptyset$
} (\lX,0) -- (qb);

\draw[tran] (qa) -- (\rX,\maY) to node[auto,sloped,above] {
    $\{ \beta \}, y \le W_\beta, \{ y \}$
} (\rX,\mbY) -- (qb);

\draw[tran] (qb) -- (\lX,\mbY) to node[auto,sloped,above] {
    $\{ \alpha \}, y \le W_\alpha, \{ y \}$
} (\lX,\maY) -- (qa);



\end{tikzpicture}
            }
        \caption{A DTRA Specification}
        \label{fig:dta}
    \end{minipage}
\end{figure}  


In the following example, we illustrate a PTA which models a simple task-processing example.

\begin{example}\label{ex:pta}
Consider the PTA depicted in Figure~\ref{fig:pta}.
$\WORK{\alpha},\WORK{\beta}$ are locations
and $x$ is the only clock. Below each location first comes (vertically) its invariant condition and then the set of labels assigned to the location. For example, $\inv(\WORK{\alpha})=x \le 10$
and $\lbfunc(\WORK{\alpha})=\{ \alpha \}$.
The two dots together with their corresponding solid line and dashed arrows refer to two actions $\tau_\alpha,\tau_\beta$ with their enabling conditions and transition probabilities given by the probabilistic transition function.
For example, the upper dot at the right of $\WORK{\alpha}$ refers to the action $\tau_\alpha$ for which $\enab(\WORK{\alpha}, \tau_\alpha)=\true$, $\prob(\WORK{\alpha}, \tau_\alpha)(\{x\},\WORK{\alpha})=0.1$, and $\prob(\WORK{\alpha}, \tau_\alpha)(\{x\},\WORK{\beta})=0.9$.
The PTA models a faulty machine which processes two different kinds of jobs (i.e., $\alpha,\beta$) in an alternating fashion.
If the machine fails to complete the current job, then it will repeat processing the job until it completes the job.
For job $\alpha$, the machine always processes the job within $10$ time units (cf. the invariant condition $x\le 10$), but may fail to complete the job with probability $0.1$;
Analogously, the machine always processes the job $\beta$ within $15$ time units (cf. the invariant condition $x\le 15$), but may fail to complete the job with probability $0.2$.
Note that we omit the initial location in this example.
\end{example}
\vspace{-1.5em}
\subsection{Timed Automata}
\vspace{-0.6em}
\begin{definition}[Timed Automata ~\cite{DBLP:journals/tse/DonatelliHS09,DBLP:journals/corr/abs-1101-3694,DBLP:conf/hybrid/Fu13}]
A \emph{timed automaton} (TA)
$\nta$ is a tuple
\begin{equation}\label{eq:nta}
\nta=(\cstates,\alphabet,\dtclocks,\rules)
\end{equation}
where
\begin{compactitem}
\item $\cstates$ is a finite set of \emph{modes}; 
\item $\alphabet$ is a finite \emph{alphabet} of \emph{symbols} disjoint from $[0,\infty)$;
\item $\dtclocks$ is a finite set of \emph{clocks};
\item $\rules\subseteq \cstates\times\alphabet\times\clcons{\dtclocks}\times 2^{\dtclocks}\times \cstates$ is a finite set of \emph{rules}.
\end{compactitem}
$\dta$ is a \emph{deterministic} TA (DTA) if the following holds:
\begin{compactenum}
\item ({\em determinism}) for $(\dtloc_i,b_i,\phi_i,X_i,\dtloc'_i)\in\rules$ ($i\in\{1,2\}$), if $(\dtloc_1,b_1)=(\dtloc_2,b_2)$ and $\sat{\phi_1}\cap \sat{\phi_2}\ne\emptyset$ then $(\phi_1,X_1,\dtloc'_1)=(\phi_2,X_2,\dtloc'_2)$;
\item ({\em totality}) for all $(q,b)\in \cstates\times\Sigma$ and $\nu\in\val{\clocks}$, there exists $(q,b,\phi,X,q')\in\rules$ such that $\nu\models \phi$.
\end{compactenum}
\end{definition}
Informally, A TA is deterministic if there is always exactly one rule applicable for the timed transition. 
We do not incorporate invariants in TAs as we use TAs as language acceptors. 

Below we illustrate the semantics of TAs. We fix a TA $\dta$ in the form (\ref{eq:nta}).


\smallskip\noindent{\em Configurations and One-Step Transition Relation.}
A \emph{configuration} is a pair $(\dtloc,\nu)$, where $\dtloc\in \cstates$ and $\nu\in\val{\dtclocks}$.
The \emph{one-step transition relation}
\[
\Rightarrow~\subseteq~ (\cstates \times \val{\dtclocks}) \times (\alphabet\cup [0,\infty)) \times (\cstates \times \val{\dtclocks})
\]
is defined by: $((\dtloc,\nu),a,(\dtloc',\nu'))\in\Rightarrow$ iff
either (i) $a\in [0,\infty)$ and $(\dtloc',\nu')=(\dtloc,\nu+a)$
or (ii) $a\in\alphabet$ and there exists a rule $(q,a,\phi,X,q')\in\Delta$
such that $\nu\models\phi$ and $\nu'=\reset{\nu}{X}$.
For the sake of convenience, we write $\dtatr{(\dtloc,\nu)}{a}{(\dtloc',\nu')}$
instead of $((\dtloc,\nu),a,(\dtloc',\nu'))\in\Rightarrow$.
Note that if $\dta$ is deterministic, then there is a unique $(\dtloc',\nu')$ such that
$\dtatr{(\dtloc,\nu)}{a}{(\dtloc',\nu')}$ given any $(\dtloc,\nu), a$.

\smallskip
\noindent{\em Infinite Timed Words and Runs.}
An \emph{infinite timed word} is an infinite sequence $w=\{a_n\}_{n\in\Nset_0}$ such that
$
    a_{2n}      \in [0,\infty)
    \mbox{ and }
    a_{2n+1}    \in \alphabet
$
for all $n$;
the infinite timed word $w$ is \emph{time-divergent} if $\sum_{n\in\Nset_0}a_{2n}=\infty$.
A \emph{run} of $\dta$ on an infinite timed word $w=\{a_n\}_{n\in\Nset_0}$
with \emph{initial configuration} $(\dtloc,\nu)$,
is an infinite
sequence $\xi=\{\left(\dtloc_n,\nu_n,a_n\right)\}_{n\in\Nset_0}$
satisfying that $(\dtloc_0,\nu_0)=(\dtloc,\nu)$ and $\dtatr{(\dtloc_n,\nu_n)}{a_n}{(\dtloc_{n+1},\nu_{n+1})}$
for all $n\in\Nset_0$;
the \emph{trajectory}
$\traj{ \xi }$
of the run
$\xi$
is defined as an infinite word over $\cstates$ such that
$\traj{ \xi } := q_0 q_1 \dots$~.
Note that if $\dta$ is deterministic, then there is a unique run on every infinite timed word.

Below we illustrate the acceptance condition for TAs.
We consider Rabin acceptance condition as the infinite acceptance condition.


\vspace{-0.5em}
\begin{definition}[Rabin Acceptance Condition \cite{DBLP:books/daglib/0020348}]
A TA with \emph{Rabin acceptance condition} (TRA) is a tuple
\begin{equation}\label{eq:tra}
    \dta=(\cstates,\alphabet,\dtclocks,\rules,\rabin)
\end{equation}
where $(\cstates,\alphabet,\dtclocks,\rules)$ is a TA and $\rabin$ is a finite
set of pairs
$
    \rabin
        = \{
            (H_1,K_1 ),
            \dots,
            (H_n,K_n)
        \}
$ representing a Rabin condition for which
$H_i$ and $K_i$ are subsets of $\cstates$ for all $i\le n$.
$\dta$ is a \emph{deterministic} TRA (DTRA) if $(\cstates,\alphabet,\dtclocks,\rules)$ is a DTA.
A set $\cstates' \subseteq \cstates $ is \emph{Rabin-accepting} by $\rabin$,
written as the predicate $\accept{\cstates'}{\rabin}$,
if there is $ 1 \leq i \leq n$ such that $ \cstates' \cap H_i= \emptyset $
and $ \cstates' \cap K_i \neq \emptyset $. An infinite timed word $w$ is \emph{Rabin-accepted} by
$\dta$ with \emph{initial configuration} $(\dtloc,\nu)$ iff there exists a run $\xi$ of $(\cstates,\alphabet,\dtclocks,\rules)$ on $w$ with $(\dtloc,\nu)$ such that
$
    \infset{
         \traj{\xi}
    }
$ is Rabin-accepting by $\rabin$.
\end{definition}


\begin{example}\label{ex:dta}
Consider the DTRA depicted in Figure~\ref{fig:dta}.
The alphabet of this DTRA is the powerset of atomic propositions in Figure~\ref{fig:pta}.
In the figure, $\INIT,\q{\alpha},\q{\beta}$ and $\FAIL$ are modes with the Rabin condition
$
    \rabin
    =
    \{
        (\{ \FAIL \},
        \{
            \q{\alpha},
            \q{\beta}
        \})
    \}
$, $y$ is a clock and arrows between modes are rules.
$C_\gamma,W_\gamma$ ($\gamma\in\{\alpha,\beta\}$) are undetermined integer constants.
For example, there are four rules emitting from $\q{\alpha}$:
\begin{align*}
(\q{\alpha}, \{ \alpha \}, y \le C_\alpha,\emptyset, \q{\alpha}),& ~(\q{\alpha}, \{ \beta \}, y \le W_\beta, \{ y \}, \q{\beta}),\\
(\q{\alpha}, \{ \alpha \}, C_\alpha<y,\emptyset, \FAIL),& ~(\q{\alpha}, \{ \beta \}, W_\beta<y, \emptyset, \FAIL).
\end{align*}
$\INIT$ is the initial mode to read the first symbol upon which transiting to either $\q{\alpha}$ or $\q{\beta}$.
$\FAIL$ is a deadlock mode from which all rules go to itself.
Note that the rules of the DTRA does not satisfy the totality condition.
However, we assume that all missing rules lead to the mode $\FAIL$ and does not affect the Rabin acceptance condition.
The mode $\q{\alpha}$ does not reset the clock $y$ until it reads $\beta$.
Moreover, $\q{\alpha}$ does not transit to $\FAIL$ only if the time spent within a maximal consecutive segment of $\alpha$'s (in an infinite timed word) is no greater than $C_\alpha$ time units (cf. the rule $(\q{\alpha}, \{ \alpha \}, y \le C_\alpha,\emptyset, \q{\alpha})$) and the total time from the start of the segment until $\beta$ is read (the time within a maximal consecutive segment of $\alpha$'s plus the time spent on the last $\alpha$ in the segment) is no greater than $W_\beta$
(cf. the rule $(\q{\alpha}, \{ \beta \}, y \le W_\beta, \{ y \}, \q{\beta})$).
The behaviour of the mode $\q{\beta}$ can be argued similar to that of $\q{\alpha}$ 
where the only difference is to flip $\alpha$ and $\beta$.
From the Rabin acceptance condition, the DTRA specifies a property on infinite timed words that the time spent within a maximal consecutive segment of $\alpha$'s (resp. $\beta$'s) and the total time until $\beta$ (resp. $\alpha$) is read
always satisfy the conditions specified by $\q{\alpha}$ (resp. $\q{\beta}$).
\end{example}

\vspace{-1.8em}
\section{Problem Statement}
\vspace{-1em}
In this part, we define the {\sc PTA-TRA} problem of model-checking {PTAs} against TA-specifications.
The problem takes a PTA and a TRA as input, and computes the minimum and the maximum probability that infinite paths under the PTA are accepted by the TRA.
Informally, the TRA encodes the linear dense-time property by judging whether an infinite path is accepted or not through its external behaviour,
then the problem is to compute the probability that an infinite path is accepted by the TRA.
In practice, the TRA can be used to capture all good (or bad) behaviours, so the problem can be treated as a task to evaluate to what extent the PTA behaves in a good (or bad) way.

Below we fix a well-formed PTA $\pta$ taking the form (\ref{eq:pta}) and a TRA $\dta$ taking the form (\ref{eq:tra}).
W.l.o.g., we assume that $\clocks\cap\dtclocks=\emptyset$ and $\alphabet=2^{\ap}$.
We first show how an infinite path in $\infpaths{\pta}$ can be interpreted as an infinite timed word.
%
\begin{definition}[Infinite Paths as Infinite Timed Words]\label{def:interpretation}
Given an infinite path
$
    \infpath
        =
            (\loc_0,\nu_0)
            a_0
            (\loc_1,\nu_1)
            a_1
            (\loc_2,\nu_2)a_2
            \dots
$
under $\pta$, the infinite timed word $\lbfunc(\infpath)$
is defined as
$
\lbfunc(\infpath):=a_0\lbfunc(\loc_2)a_2\lbfunc(\loc_4)\dots a_{2n}\lbfunc(\loc_{2n+2}) \dots\enskip.
$
Recall that $\nu_0=\zero$, $a_{2n}\in [0,\infty)$ and $a_{2n+1}\in\acts$ for $n\in\Nset_0$.
\end{definition}
\begin{remark}
Informally, the interpretation in Definition~\ref{def:interpretation} works
by (i) dropping (a) the initial location $\loc_0$, (b) all clock valuations $\nu_n$'s,
(c) all locations $\loc_{2n+1}$'s following a time-elapse,
(d) all internal actions $a_{2n+1}$'s of $\pta$ and (ii) replacing every $\loc_{2n}$ ($n\ge 1$) by $\lbfunc(\loc_{2n})$.
The interpretation captures only external behaviours including time-elapses and labels of locations upon state-change, and discards internal behaviours such as the concrete locations, clock valuations and actions.
Although the interpretation ignores the initial location,
we deal with it in our acceptance condition where the initial location is preprocessed by the TRA.
\end{remark}
%
\vspace{-1em}
\begin{definition}[Path Acceptance]\label{def:fnacc}
An infinite path $\infpath$ of $\pta$ is \emph{accepted} by $\dta$ w.r.t
initial configuration $(\dtloc,\nu)$, written as the single predicate $\acccept{\tra}{(\dtloc,\nu)}{\infpath}$,
if there is a configuration $(\dtloc',\nu')$ such that $\dtatr{(\dtloc,\nu)}{\lbfunc(\loc^*)}{(\dtloc',\nu')}$ and
the infinite word $\lbfunc(\infpath)$ is Rabin-accepted by $\dta$ with
$
\left(q',\nu'\right)
$.
\end{definition}

The initial location omitted in Definition~\ref{def:interpretation} is preprocessed by specifying explicitly that the first label $\lbfunc(\loc^*)$ is read by the initial configuration $(\dtloc,\nu)$.
Below we define acceptance probabilities over infinite paths under $\pta$.
\vspace{-0.3em}
\begin{definition}[Acceptance Probabilities]\label{def:accprob}
The probability that $\pta$ \emph{observes} $\tra$ under scheduler $\sigma$ and initial mode $\dtloc\in\cstates$, denoted by $\pr{\dtloc}{\sigma}$, is defined by:
\[
    \pr{\dtloc}{\sigma}
        :=
            \probm^{\pta,\sigma}(
                \LangCsAqF
            )
\]
where $\LangCsAqF$ is the set of infinite paths under $\pta$ that are accepted by the TRA $\tra$ w.r.t $(\dtloc,\zero)$ i.e.
$
    \LangCsAqF = \left \{
        \infpath \in \infpaths{\pta,\sigma} \mid
        \acccept
            {\tra}
            {(\dtloc,\zero)}
            {\infpath}
    \right\}.
$
\end{definition}
Since the set $\fnpaths{\pta,\sigma}$ is countably-infinite,
$\LangCsAqF$ is measurable since it can be represented as a countable intersection of certain countable unions of some cylinder sets (cf.~\cite[Remark 10.24]{DBLP:books/daglib/0020348} and Appendix \ref{app:proof} for details).


Now we introduce the {\sc PTA-TRA} problem.

\begin{compactitem}
\item {\bf Input:} a well-formed PTA $\pta$, a TRA  $\dta$ and an initial mode $\dtloc$ in $\dta$;
\item {\bf Output:} $\inf_\sigma \pr{\dtloc}{\sigma}$ and $\sup_\sigma  \pr{\dtloc}{\sigma}$, where $\sigma$ ranges over all time-divergent schedulers for $\pta$.
\end{compactitem}

We refer to the problem as {\sc PTA-DTRA} if $\dta$ is deterministic.




\vspace{-1.6em}
\section{The PTA-DTRA Problem}
\vspace{-1em}
In this section, we solve the {\sc PTA-DTRA} problem through a product construction.
Based on the product construction, we also settle the complexity of the problem.
Below we fix a well-formed PTA $\pta$ in the form (\ref{eq:pta}) and a DTRA $\dta$ in the form (\ref{eq:tra}).
W.l.o.g, we consider that $\clocks\cap\dtclocks=\emptyset$ and $\alphabet=2^{\ap}$.

\smallskip
\noindent {\bf The Main Idea.} The core part of the product construction is a PTA which preserves the probability of the set of infinite paths accepted by $\dta$.
The intuition is to let $\dta$ reads external actions of $\pta$ while $\pta$ evolves along the time axis.
The major difficulty is that when $\pta$ performs actions in $\acts$, there is a probabilistic choice between the target locations. Then $\dta$ needs to know the labelling of the target location and the rule (in $\rules$) used for the transition.
A naive solution is to integrate each single rule in $\rules$ into the enabling condition $\enab$ in $\pta$. However, this simple solution does not work since a single rule fixes the labelling of a location in $\pta$, while the probability distribution (given by $\prob$) can jump to locations with different labels.
We solve this difficulty by integrating into the enabling condition
enough information on clock valuations under $\dta$ so that the rule used for the transition is clear.

\smallskip
\noindent{\textbf{The Product Construction.}}
For each $\dtloc\in\cstates$, we let
\begin{align*}
\exttuples_{\dtloc}:=\{h:\alphabet\rightarrow\clcons{\dtclocks}\mid
\forall b\in\alphabet.\left(\dtloc, b, h(b), X, \dtloc')\in\rules\mbox{ for some }X, \dtloc'\right)\}\enskip.
\end{align*}
The totality of $\rules$ ensures that $\exttuples_{\dtloc}$ is non-empty.
Intuitively, every element of $\exttuples_{\dtloc}$ is a tuple of clock constraints $\{\phi_b\}_{b\in\alphabet}$, where each
clock constraint $\phi_b$ is chosen from the rules emitting from $\dtloc$ and $b$.
The \emph{product PTA} $\product{\pta}{\dta_q}$ (between $\pta$ and $\dta$ with initial mode $\dtloc$) is defined as
$
\left({\locs}_\otimes, \loc^*_\otimes, \clocks_\otimes, \acts_\otimes, \inv_\otimes, \enab_\otimes,  \prob_\otimes, \cstates, \lbfunc_\otimes\right)
$
where :

\begin{compactitem}
\item $\locs_\otimes:=\locs\times \cstates$;
\item $\loc^*_\otimes:=(\loc^*, q^\star)$ where $q^\star$ is the unique mode such that $\dtatr{(\dtloc,\zero)}{\lbfunc(\loc^*)}{(q^\star,\zero)}$;
\item $\clocks_\otimes:=\clocks\cup\dtclocks$;
\item $\acts_\otimes:=\acts\times\bigcup_\dtloc\exttuples_\dtloc$; 
\item $\inv_\otimes(\loc,\dtloc):=\inv(\loc)$ for all $(\loc,\dtloc)\in \locs_\otimes$;
\item $\enab_\otimes\left((\loc,\dtloc), (a,h)\right):=\enab(\loc,a)\wedge \bigwedge_{b\in\alphabet} h(b)$ if $h\in \exttuples_\dtloc$,and $\enab_\otimes\left((\loc,\dtloc), (a,h)\right):=\false$ otherwise, for all $(\loc,\dtloc)\in \locs_\otimes$, $(a,h) \in  \acts_\otimes$.
\item
$
    \lbfunc_\otimes \left(
        \loc,\dtloc
    \right)
        := \left \{
            \dtloc
        \right \}
    \mbox{ for all } \left (
        \loc,\dtloc
    \right )
    \in \locs_\otimes;
$
\item $\prob_\otimes$ is given by
\begin{align*}
&\prob_\otimes\left((\loc,\dtloc),(a,h)\right)(Y,(\loc',\dtloc')):=\\
&~~\begin{cases}
\prob\left(\loc,a\right)(Y\cap \clocks,\loc') & \mbox{if } (\dtloc,\lbfunc\left(\loc'\right), h(\lbfunc\left(\loc'\right)), Y\cap \dtclocks, \dtloc')\\
& \mbox{\quad is a (unique) rule in }\rules\\
0 & \mbox{otherwise}
\end{cases}\enskip.
\end{align*}
for all $(\loc,\dtloc),(\loc',\dtloc')\in \locs_\otimes$, $(a,h) \in  \acts_\otimes$ , $Y \in \clocks_\otimes$
\end{compactitem}
Besides standard constructions (e.g., the Cartesian product between $\locs$ and $\cstates$), the product construction also has Cartesian product between $\acts$ and $\bigcup_\dtloc\exttuples_\dtloc$. For each extended action $(a,h)$, the enabling condition for this action is the conjunction between $\enab(\loc,a)$ and all clock constraints from $h$.
This is to ensure that when the action $(a,h)$ is taken, the clock valuation under $\dta$ satisfies every clock constraint in $h$.
Then in the definition for $\prob_\otimes$, upon the action $(a,h)$, the product PTA first perform probabilistic jump from $\pta$ with the target location $\loc'$, then chooses the unique rule $(\dtloc,\lbfunc\left(\loc'\right), h(\lbfunc\left(\loc'\right)), Y\cap \dtclocks, \dtloc')$ from the emitting mode $\dtloc$ and the label $\lbfunc\left(\loc'\right)$ for which the uniqueness comes from the determinism of $\rules$, then perform the discrete transition from $\dta$.
Finally, we label each $(\loc,\dtloc)$ by $\dtloc$ to meet the Rabin acceptance condition.\qed

It is easy to see that the PTA $\product{\pta}{\dta_q}$ is well-formed as $\pta$ is well-formed and $\dta$ does not introduce extra invariant conditions.
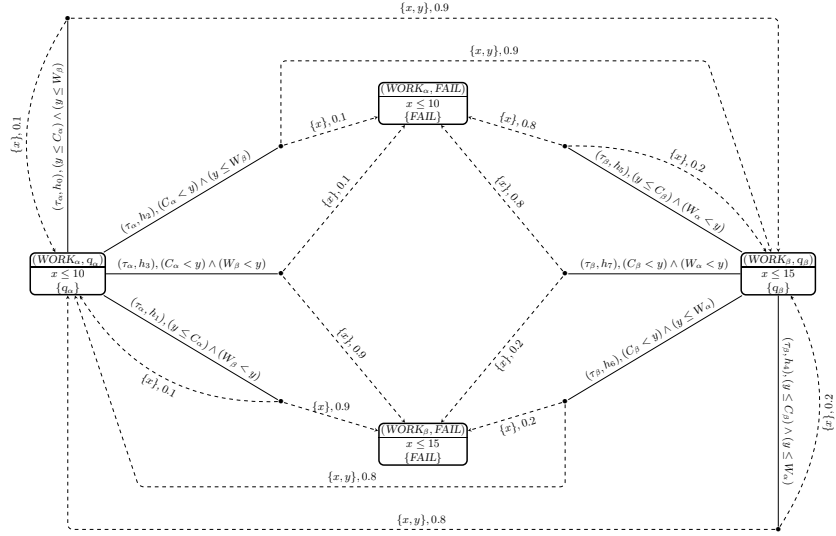
\begin{figure*}
    \centering
    \resizebox{1.0\textwidth}{!}{
        \def \hY { 6}
\def \hhY{ 5}
\def \aY { 4}
\def \aaY{ 3}
\def \abY{ 1}
\def \mY { 0}
\def \acY{-1}
\def \adY{-3}
\def \bY {-4}
\def \llY{-5}
\def \lY {-6}

\def \aX    {-10}
\def \aaX   { -8}
\def \adecX { -4}
\def \fX    {  0}
\def \bdecX {  4}
\def \bbX   {  8}
\def \bX    { 10}

\def \rX {16}

\begin{tikzpicture} [y =  1.2 cm]

\location{\PairV{a}{qa}}{(\aX,\mY)}
    {
        \PairS
            {\WORK{\alpha}}
            {\q{\alpha}}
    }
    { x\le 10 }
    { \{ \q{\alpha} \} };

\location{\PairV{b}{qb}}{(\bX,\mY)}
    {
        \PairS
            {\WORK{\beta}}
            {\q{\beta}}
    }
    { x\le 15 }
    { \{ \q{\beta} \} };

\location{\PairV{a}{f}}{(\fX,\aY)}
    {
        \PairS
            {\WORK{\alpha}}
            {\FAIL}
    }
    { x\le 10 }
    { \{ \FAIL \} };

\location{\PairV{b}{f}}{(\fX,\bY)}
    {
        \PairS
            {\WORK{\beta}}
            {\FAIL}
    }
    { x\le 15 }
    { \{ \FAIL \} };


\node[dec] (aadec) at (\aX,\hY) {$\bullet$};
\draw[nchoice] (\PairV{a}{qa}) to node[auto,above,sloped] {
    $(\tau_\alpha,h_0),( y \le C_\alpha ) \land ( y \le W_\beta )$
} (aadec);
\draw[pchoice,bend right] (aadec) to node[auto,above,sloped] {
    $ \{ x \}, 0.1$
} (\PairV{a}{qa});
\draw[pchoice] (aadec) to node[auto,above] {
    $ \{ x, y \},0.9$
} (\bX,\hY) -- (\PairV{b}{qb});

\node[dec] (abdec) at (\adecX,\adY) {$\bullet$};
\draw[nchoice] (\PairV{a}{qa}) to node[auto,above,sloped] {
    $(\tau_\alpha,h_1),( y \le C_\alpha ) \land ( W_\beta < y )$
} (abdec);
\draw[pchoice,bend left] (abdec) to node[auto,below,sloped] {
    $ \{ x \}, 0.1$
} (\PairV{a}{qa});
\draw[pchoice] (abdec) to node[auto,above] {
    $ \{ x \}, 0.9$
} (\PairV{b}{f});

\node[dec] (acdec) at (\adecX,\aaY) {$\bullet$};
\draw[nchoice] (\PairV{a}{qa}) to node[auto,above,sloped] {
    $(\tau_\alpha,h_2),( C_\alpha < y ) \land ( y \le W_\beta )$
} (acdec);
\draw[pchoice] (acdec) to node[auto,above,sloped] {
    $ \{ x \}, 0.1$
} (\PairV{a}{f});
\draw[pchoice] (acdec) -- (\adecX,\hhY) to node[auto,above] {
    $ \{ x, y \}, 0.9$
} (\bbX,\hhY) -- (\PairV{b}{qb});

\node[dec] (addec) at (\adecX,\mY) {$\bullet$};
\draw[nchoice] (\PairV{a}{qa}) to node[auto,above,sloped] {
    $(\tau_\alpha,h_3),( C_\alpha < y ) \land ( W_\beta < y )$
} (addec);
\draw[pchoice] (addec) to node[auto,above,sloped] {
    $ \{ x \}, 0.1$
} (\PairV{a}{f});
\draw[pchoice] (addec) to node[auto,above,sloped] {
    $ \{ x \}, 0.9$
} (\PairV{b}{f});


\node[dec] (badec) at (\bX,\lY) {$\bullet$};
\draw[nchoice] (\PairV{b}{qb}) to node[auto,above,sloped] {
    $(\tau_\beta,h_4),( y \le C_\beta ) \land ( y \le W_\alpha )  $
} (badec);
\draw[pchoice,bend right] (badec) to node[auto,below,sloped] {
    $ \{ x \}, 0.2$
} (\PairV{b}{qb});
\draw[pchoice] (badec) to node[auto,above,sloped] {
    $ \{ x, y \},0.8$
} (\aX,\lY) -- (\PairV{a}{qa});

\node[dec] (bbdec) at (\bdecX,\aaY) {$\bullet$};
\draw[nchoice] (\PairV{b}{qb}) to node[auto,above,sloped] {
    $(\tau_\beta,h_5),( y \le C_\beta )  \land ( W_\alpha < y )$
} (bbdec);
\draw[pchoice,bend left] (bbdec) to node[auto,above,sloped] {
    $ \{ x \}, 0.2$
} (\PairV{b}{qb});
\draw[pchoice] (bbdec) to node[auto,above,sloped] {
    $ \{ x \}, 0.8$
} (\PairV{a}{f});

\node[dec] (bcdec) at (\bdecX,\adY) {$\bullet$};
\draw[nchoice] (\PairV{b}{qb}) to node[auto,above,sloped] {
    $(\tau_\beta,h_6),( C_\beta < y ) \land ( y \le W_\alpha )$
} (bcdec);
\draw[pchoice] (bcdec) to node[auto,below,sloped] {
    $ \{ x \}, 0.2$
} (\PairV{b}{f});
\draw[pchoice] (bcdec) -- (\bdecX,\llY) to node[auto,above,sloped] {
    $ \{ x, y \}, 0.8$
} (\aaX,\llY) --(\PairV{a}{qa});

\node[dec] (bddec) at (\bdecX,\mY) {$\bullet$};
\draw[nchoice] (\PairV{b}{qb}) to node[auto,above,sloped] {
    $(\tau_\beta,h_7),( C_\beta < y ) \land ( W_\alpha < y )$
} (bddec);
\draw[pchoice] (bddec) to node[auto,below,sloped] {
    $ \{ x \}, 0.2$
} (\PairV{b}{f});
\draw[pchoice] (bddec) to node[auto,above,sloped] {
    $ \{ x \}, 0.8$
} (\PairV{a}{f});

\end{tikzpicture}
        }
    \caption{The Product PTA for Our Running Example}
    \label{fig:product}
\end{figure*}
\begin{example}
The product PTA between the PTA in Example~\ref{ex:pta} and the DTRA in Example~\ref{ex:dta} is depicted in Figure~\ref{fig:product}.
In the figure, $\PairS{\WORK{\alpha}}{\q{\alpha}}, \PairS{\WORK{\beta}}{\q{\beta}}$ and
$\PairS{\WORK{\alpha}}{\FAIL}, \PairS{\WORK{\beta}}{\FAIL}$
are product locations. We omit the initial location and unreachable locations in the product construction.
From the construction of $\exttuples_{\dtloc}$'s, the functions $h_i$'s are as follows (we omit redundant labels such as $\emptyset$ and $\{\alpha,\beta\}$ which never appear in the PTA):
\begin{compactitem}
\item $h_0=\{\{ \alpha \} \mapsto y \le C_\alpha ,\{ \beta  \} \mapsto y \le W_\beta\}$;
\item $h_1=\{\{ \alpha \} \mapsto y \le C_\alpha ,\{ \beta  \} \mapsto W_\beta< y\}$;
\item $h_2=\{\{ \alpha \} \mapsto C_\alpha < y ,\{ \beta  \} \mapsto y \le W_\beta\}$;
\item $h_3=\{\{ \alpha \} \mapsto C_\alpha < y ,\{ \beta  \} \mapsto W_\beta< y\}$;
\item $h_4=\{\{ \beta \} \mapsto y \le C_\beta ,\{ \alpha  \} \mapsto y\le W_\alpha\}$;
\item $h_5=\{\{ \beta \} \mapsto y \le C_\beta ,\{ \alpha  \} \mapsto W_\alpha< y\}$;
\item $h_6=\{\{ \beta \} \mapsto C_\beta < y ,\{ \alpha  \} \mapsto y \le W_\alpha\}$;
\item $h_7=\{\{ \beta \} \mapsto C_\beta < y ,\{ \alpha  \} \mapsto W_\alpha< y\}$.
\end{compactitem}
The intuition is that the DTA accepts all infinite paths under the PTA such that the failing time for job $\gamma$ ($\gamma\in\{\alpha,\beta\}$) (the time within the consecutive $\gamma$'s) should be no greater than $C_\gamma$ and the waiting time for job $\gamma$ (the failing time plus the time spent on the last $\gamma$) should be no greater than $W_\gamma$.
\end{example}


Below we clarify the correspondence between $\pta,\dta$ and $\product{\pta}{\dta_\dtloc}$.
We first show the relationship between paths under $\pta$ and those under $\product{\pta}{\dta_\dtloc}$.
Informally, paths under $\product{\pta}{\dta_\dtloc}$ are just paths under $\pta$ extended with runs of $\dta$.

\smallskip\noindent
{\textbf{Transformation $\pfunc$ for Paths from $\pta$ into $\product{\pta}{\dta_\dtloc}$.}}
The transformation is defined as the function $\pfunc:\fnpaths{\pta}\cup\infpaths{\pta}\rightarrow \fnpaths{\product{\pta}{\dta_\dtloc}}\cup\infpaths{\product{\pta}{\dta_\dtloc}}$
which transform a finite or infinite path under $\pta$ into one under  $\product{\pta}{\dta_\dtloc}$.
For a finite path
$
\fnpath=(\loc_0,\nu_0)a_0\dots a_{n-1}(\loc_n,\nu_n)
$
under $\pta$ (note that $(\loc_0,\nu_0)=(\loc^*, \zero)$ by definition),
we define $\pfunc(\fnpath)$ to be the unique finite path
\begin{equation}\label{eq:trho}
\pfunc(\fnpath):=((\loc_0,\dtloc_0),\nu_0\cup\mu_0)a'_0\dots a'_{n-1}((\loc_n,\dtloc_n),\nu_n\cup\mu_n)
\end{equation}
under $\product{\pta}{\dta_\dtloc}$ such that the following conditions (\dag) hold:
\begin{compactitem}
\item $\dtatr{(\dtloc,\zero)}{\lbfunc(\loc^*)}{(q_0,\mu_0)}$ (note that $\mu_0=\zero$);
\item for all $0\le k< n$, if $a_k\in [0,\infty)$ then $a'_k=a_k$ and $\dtatr{(\dtloc_k,\mu_k)}{a_k}{\left(\dtloc_{k+1},\mu_{k+1}\right)}$;
\item for all $0\le k< n$, if $a_k\in\acts$ then $a'_k=(a_k,\xi_k)$ and $\dtatr{(\dtloc_k,\mu_k)}{\lbfunc(\loc_{k+1})}{\left(\dtloc_{k+1},\mu_{k+1}\right)}$ where $\xi_k$ is the unique function such that for each symbol $b\in\alphabet$, $\xi_k(b)$ is the unique clock constraint appearing in a rule emitting from $q_k$ and with symbol $b$ such that $\mu_k\models\xi_k(b)$.
\end{compactitem}
Likewise, for an infinite path $\infpath=(\loc_0,\nu_0)a_0(\loc_1,\nu_1)a_1\dots$
under $\pta$, we define $\pfunc(\infpath)$ to be the unique infinite path
\begin{equation}\label{eq:trinfpath}
\pfunc(\infpath):=((\loc_0,\dtloc_0),\nu_0\cup\mu_0)a'_0((\loc_1,\dtloc_1),\nu_1\cup\mu_1)a'_1\dots
\end{equation}
under $\product{\pta}{\dta_\dtloc}$ such that the three conditions in (\dag) hold for all $k\in\Nset_0$ instead of all $0\le k< n$. From the determinism and totality of $\dta$, it is straightforward to prove the following result.
\begin{lemma}\label{lemm:pfuncbij}
The function $\pfunc$ is a bijection. Moreover, for any infinite path $\infpath$ under $\pta$, $\infpath$ is non-zeno iff $\pfunc(\infpath)$ is non-zeno.
\end{lemma}


Below we also show the correspondence on schedulers before and after the product construction.

\smallskip
\noindent{\textbf{Transformation $\sfunc$ for Schedulers from $\pta$ into $\product{\pta}{\dta_\dtloc}$.}}
We define the function $\sfunc$ from the set of schedulers under $\pta$ into the set of schedulers under $\product{\pta}{\dta_\dtloc}$ as follows: for any scheduler $\sigma$ for $\pta$, $\sfunc(\sigma)$ (for $\product{\pta}{\dta_\dtloc}$) is defined such that for any finite path $\fnpath$ under $\pta$ where $\fnpath=(\loc_0,\nu_0)a_0\dots a_{n-1}(\loc_n,\nu_n)$ and $\pfunc(\rho)$ given as in (\ref{eq:trho}),
\[
\sfunc(\sigma)(\pfunc(\fnpath)):=
\begin{cases}
\sigma(\fnpath) & \mbox{if }n\mbox{ is even} \\
(\sigma(\fnpath),\lambda(\rho)) & \mbox{if }n\mbox{ is odd}
\end{cases}
\]
where $\lambda(\rho)$ is
the unique function such that for each symbol $b\in\alphabet$, $\lambda(\rho)(b)$ is the clock constraint in the unique rule emitting from $q_n$ and with symbol $b$ such that $\mu_n\models\lambda(\rho)(b)$.
Note that the well-definedness of $\sfunc$ follows from Lemma~\ref{lemm:pfuncbij}.

From Lemma~\ref{lemm:pfuncbij}, the product construction, the determinism and totality of $\rules$, one can prove directly the following lemma.
\begin{lemma}\label{lemm:sfuncbij}
The function $\sfunc$ is a bijection.
\end{lemma}
Now we prove the relationship between infinite paths accepted by a DTRA before product construction and infinite paths satisfying certain Rabin condition.

We introduce more notations.
First, we lift the function $\pfunc$ to all subsets of paths in the standard fashion: for all subsets $A\subseteq \fnpaths{\pta}\cup\infpaths{\pta}$, $\pfunc(A):=\{\pfunc(\omega)\mid \omega\in A\}$.
Then for an infinite path $\infpath$ under $\product{\pta}{\dta_\dtloc}$ in the form
(\ref{eq:trinfpath}), we define
the \emph{trace} of $ \infpath  $ as an infinite word over $\cstates$ by
$\trace{  \infpath  } := \dtloc_0 \dtloc_1 \dots $ .
Finally, for any scheduler $\sigma$ for $\product{\pta}{\dta_\dtloc}$,
we define the set $\rabinp{\sigma}$ by
$$
    \rabinp{\sigma}:=\left \{
        \infpath \in \infpaths{\product{\pta}{\dta_\dtloc},\sigma} \mid
        \accept{
            \infset{
                \trace{
                    \infpath
                }
            }
        }   {
            \rabin
        }
    \right\}.
$$
Intuitively, $\rabinp{\sigma}$ is the set of infinite paths under $\product{\pta}{\dta_\dtloc}$ that meet the Rabin condition $\rabin$ from $\dta$.
The following proposition clarifies the role of $\rabinp{\sigma}$.
%
\begin{proposition}\label{prop:psfunc}
For any scheduler $\sigma$ for $\pta$ and any initial mode $q$ for $\dta$, we have $\TLang = \TAcc.$
\end{proposition}

Finally, we demonstrate the relationship between acceptance probabilities before product construction and Rabin(-accepting) probabilities after product construction.
We also clarify the probability of zenoness before and after the product construction.
The proof follows standard argument from measure theory.

\begin{proposition}\label{thm:main}
For any scheduler $\sigma$ for $\pta$ and mode $q$, the followings hold:
\begin{compactitem}
\item
{\small $
    \pr
        {\dtloc}
        {\sigma}
        =
            \probm
                ^{\pta,\sigma}
                \left(
                    \LangCsAqF
                \right)
        =
            \probm
                ^{\product{\pta}{\dta_\dtloc},\theta(\sigma)}
                \left(
                    \TAcc
                \right)
    ;
$}
\item
{\small$
    \probm
        ^{\pta,\sigma}
        \left( \{
                \infpath \mid \infpath \mbox{ is zeno}
            \}
        \right)
    =
    \probm
        ^{\product{\pta}{\dta_\dtloc},\theta(\sigma)}
        \left( \{
                \infpath' \mid \infpath' \mbox{ is zeno}
            \}
        \right).
$}
\end{compactitem}
\end{proposition}
A side result from Proposition~\ref{thm:main} says that $\sfunc$ preserves time-divergence for schedulers before and after product construction.
From Proposition~\ref{thm:main} and Lemma~\ref{lemm:sfuncbij}, one immediately obtains the following result which transforms the {\sc PTA-DTRA} problem into Rabin(-accepting) probabilities under the product PTA.

\begin{corollary}\label{crly:opt}
For any initial mode $q$,
$
    \opt_\sigma\pr{\dtloc}{\sigma}
        =
            \opt_{\sigma'}
            \probm^{\product{\pta}{\dta_\dtloc},\sigma'}\left(
            \rabinp{\sigma'}
            \right)
$
where $\opt$ refers to either $\inf$ (infimum) or $\sup$ (supremum),
$\sigma$ (resp. $\sigma'$) range over all time-divergent schedulers
for $\pta$ (resp. $\product{\pta}{\dta_\dtloc}$).
\end{corollary}

\noindent{\bf Solving Rabin Probabilities.} We follow the approach in~\cite{DBLP:conf/qest/Sproston11}
to solve Rabin probabilities over PTAs. Below we briefly describe the approach.
The approach can be divided into two steps.
The first step is to ensure time-divergence.
This is achieved by (i) making a copy for every location in the PTA, (ii) enforcing a transition
from every location to its copy to happen after $1$ time-unit elapses, (iii)
enforcing a transition from every copy location back to the original one immediately with no time-delay,
and (iv) putting a special label \textit{tick} in every copy.
Then time-divergence is guaranteed by adding the label  \textit{tick} to the Rabin condition.
The second step is to transform the problem into limit Rabin properties over MDPs~\cite[Theorem 10.127]{DBLP:books/daglib/0020348}.
This step constructs an MDP $\Region{\product{\pta}{\dta_\dtloc}}$ from the
PTA $ \product{\pta}{\dta_\dtloc} $ through a \emph{region-graph} construction
so that the problem is reduced to solving limit Rabin properties over $\Region{\product{\pta}{\dta_\dtloc}}$.
\emph{Regions} are finitely-many equivalence classes of clock valuations that serve as a finite abstraction which capture exactly reachability behaviours over timed transitions (cf.~\cite{DBLP:journals/tcs/AlurD94}).
Then standard methods based on \emph{maximal end components} (MECs) are applied to $\Region{\product{\pta}{\dta_\dtloc}}$.
In detail, the algorithm computes the reachability probability to MECs that satisfy the Rabin acceptance condition.
In order to guarantee time-divergence, the algorithm only picks up MECs with at least one location that has a \textit{tick} label.
Based on this approach, our result leads to an algorithm for solving the problem PTA-DTRA.

%

Note that in $\product{\pta}{\dta_\dtloc}$,
although the size of $\acts_\otimes$
may be exponential due to possible exponential blow-up from $\exttuples_{\dtloc}$,
one easily sees that $ | \locs_\otimes | $ is $ |\locs| \cdot |\cstates| $ and
$ | \clocks_\otimes | = | \clocks | + | \dtclocks| $.
Hence, the size of $ \Region{\product{\pta}{\dta_\dtloc}} $ is still exponential in
the sizes of $\pta$ and $\dta$.
It follows that $\opt_\sigma\pr{\dtloc}{\sigma}$ can be calculated in exponential time
from the MEC-based algorithm illustrated in~\cite[Theorem 10.127]{DBLP:books/daglib/0020348}, as is demonstrated by
the following proposition.

\begin{proposition}\label{prop:exptime}
The problem {\sc PTA-DTRA} is in EXPTIME in the size of the input PTA and DTRA.
\end{proposition}

It is proved in \cite{LaroussinieS07} that the reachablity-probability problem for arbitrary PTAs is \emph{EXPTIME}-complete.
Since Rabin acceptance condition subsumes reachability, one obtains that the problem PTA-DTRA is EXPTIME-hard (cf. Appendix~\ref{app:hardness} for details).
Thus we obtain the main result of this section which settles the computational complexity of the problem PTA-DTRA.

\begin{theorem}
The PTA-DTRA problem is EXPTIME-complete.
\end{theorem}
%
%
\begin{remark}
The main novelty for our product construction is that by adopting extended actions (i.e. $\exttuples_{\dtloc}$) and integrating them into enabling condition and probabilistic transition function, the product PTA can know which rule to use from the DTA upon any symbol to be read. This solves the problem that probabilistic jumps can lead to different locations, causing the usage of different rules from the DTA. Moreover, our product construction ensures EXPTIME-completeness of the problem.
\end{remark}

\vspace{-1.8em}
\section{The PTA-TRA problem}
\vspace{-1em}
In this section, we study the PTA-TRA problem where the input timed automaton needs not to be deterministic.
In contrast to the deterministic case (which is shown to be decidable and EXPTIME-complete in the previous section), 
we show that the problem is undecidable.

\smallskip
\noindent{\bf The Main Idea.} The main idea for the undecidability result is to reduce the universality problem of timed automata to the PTA-TRA problem. The universality problem over timed automata is well-known to be undecidable, as follows.

%
%
\begin{lemma}{(\cite[Theorem 5.2]{DBLP:journals/tcs/AlurD94})}\label{lemm:undecidability}
Given a timed automaton over an alphabet $\alphabet$ and an initial mode, the problem of deciding whether it accepts all time-divergent timed words w.r.t B\"{u}chi acceptance condition over $\alphabet$ is undecidable.
\end{lemma}
Although Lemma \ref{lemm:undecidability} is on  B\"{u}chi acceptance condition, it holds also for Rabin acceptance condition since Rabin acceptance condition extends  B\"{u}chi acceptance condition.
Actually the two acceptance conditions are equivalent over timed automata (cf.~\cite[Theorem 3.20]{DBLP:journals/tcs/AlurD94}). We also remark that Lemma \ref{lemm:undecidability} was originally for multiple initial modes, which can be mimicked by a single initial mode through aggregating all rules emitting from some initial mode as rules emitting from one initial mode. 

Now we prove the undecidability result as follows.
The proof idea is that we construct a PTA that can generate every time-divergent timed words with probability $1$ by some time-divergent scheduler.
Then the TRA accepts all time-divergent timed words iff the minimal probability that the PTA observes the TRA equals $1$.
\begin{theorem}\label{thm:traundecidability}
Given a PTA $\pta$ and a TRA $\dta$, the problem to decide whether the minimal probability
that $\pta$ \emph{observes} $\dta$ (under a given initial mode) is equal to $1$ is undecidable.
\end{theorem}
\begin{proof}[Proof Sketch]
Let $\dta=(\cstates,\alphabet,\dtclocks,\rules,\rabin)$ be any TRA where the alphabet $\alphabet = \{\ntaap{1}, \ntaap{2}, \cdots, \ntaap{k}\}$ and the initial mode is $\qstart$.
W.l.o.g, we consider that $\alphabet\subseteq 2^{\ap}$ for some finite set $\ap$.
This assumption is not restrictive since what $\ntaap{i}$'s concretely are is irrelevant, while the only thing that matters is that $\alphabet$ has $k$ different symbols.
We first construct the TRA $\dta' = (\cstates', \alphabet', \dtclocks, \rules',\rabin)$ where
$\cstates'   = \cstates  \cup \{ \qinit \}$ for which $\qinit$ is a fresh mode,
$\alphabet'  = \alphabet \cup \{ \ntaap{0} \}$ for which $\ntaap{0}$ is a fresh symbol and 
$\rules'     = \rules    \cup \{ \langle
            \qinit,
            \ntaap{0},
            \true,
            \dtclocks,
            \qstart
        \rangle
    \}$.
Then we construct the PTA :
\begin{compactitem}
    \item $\locs      :=  \alphabet'$, $\loc^*     :=  \ntaap{0} $, $\clocks    :=  \emptyset $ and $\acts      :=  \alphabet $;
    \item $\inv(\ntaap{i})              :=  \true
                                            \text{ for }
                                            \ntaap{i} \in \locs$;
    \item $\enab(\ntaap{i},\ntaap{j})   :=  \true
                                            \text{ for }
                                            \ntaap{i} \in \locs
                                            \text{ and }
                                            \ntaap{j} \in \acts$;
    \item $\prob(\ntaap{i},\ntaap{j})$ is the Dirac distribution at $(\emptyset,\ntaap{j})$ (i.e., $\prob(\ntaap{i},\ntaap{j})(\emptyset,\ntaap{j})=1$ and $\prob(\ntaap{i},\ntaap{j})(X,b)=0$ whenever $(X,b)\ne(\emptyset,\ntaap{j})$),
                                            \text{ for }
                                            $\ntaap{i} \in \locs$
                                            \text{ and }
                                            $\ntaap{j} \in \acts$;
    \item $\lbfunc(\ntaap{i})           :=  \ntaap{i}
                                            \text{ for } \ntaap{i} \in \locs$.
\end{compactitem}
Note that we allow no clocks in the construction since clocks are irrelevant for our result.
Since we omit clocks, we also treat states (of $\pta'$) as single locations.
One can prove that $\tra$ accepts all time-divergent timed words over $\Sigma$ with initial mode $\qstart$ iff
the minimal probability that $\pta'$ observes $\dta'$ with initial mode $\qinit$ equals $1$.
For details see Appendix~\ref{app:ptatraundecidability}. \qed
\end{proof}

\begin{remark}
Theorem~\ref{thm:traundecidability} shows that the problem to qualitatively decide the minimal probability is undecidable.
On the other hand, the decidability of the problem to decide maximum acceptance probabilities is left open.
\end{remark}

\vspace{-2em}
\section{Case Studies}\label{sect:casestudies}
\vspace{-0.8em}
In this section, we investigate two case studies which are simplified from real-world problems.
The first case is to complete a sequence of tasks, each having a failure probability and a processing time.
The second case is robot navigation in which a robot is given the command to reach certain destination in an area.
\vspace{-0.8em}
\subsection{Task Completion}
\vspace{-0.8em}
The {\sc Task-Completion} problem is to evaluate how well a sequence of tasks is finished.
In our setting, a task is always processed within a time frame. The exact processing time is nondeterministic.
After the processing, the task may fail to complete w.r.t certain probability.
Tasks are executed in the order where they appear in the sequence and need to be reprocessed if they fail to complete.
Example~\ref{ex:taskcompletion} illustrates a simple setting where there are only two tasks.
\vspace{-0.8em}
\begin{example}\label{ex:taskcompletion}
Consider the PTA depicted in Figure~\ref{fig:taskcompletion}.
In the figure, $\loc_i$'s ($1\le i\le 3$) are locations and $x$ is the only clock.
Below each location first comes (vertically) its invariant condition and then the set of labels assigned to the location. For example, $\inv(\loc_0)=x\le 2$ and $\lbfunc(\loc_0)=\{\alpha\}$.
The two dot points together with corresponding arrows refer to two actions and their enabling conditions and probability transition functions.
For example, the first dot at the right of $\loc_0$ refers to an action whose name is irrelevant, the enabling condition for this action (from $\loc_0$) is $1\le x\wedge x\le 2$ (cf. the dashed arrow emitting from $\loc_0$),
and the probability distribution for this action is to reset $x$ and go to $\loc_1$ with probability $0.9$ and
to reset $x$ and go to $\loc_0$ with probability $0.1$.
The PTA models a sequential completion of two tasks, where the atomic propositions $\alpha,\beta$ are used to distinguish adjacent tasks in sequential order. For the first task (indicated by the location $\loc_0$, the PTA can complete it with probability $0.9$, and the processing time is always between $1$ and $2$ time units. For the second task (indicated by the location $\loc_1$), the PTA completes it with probability $0.8$, and the completion time is always between $2$ and $3$ time units. The location $\loc_2$ signifies that all the tasks are completed.
We omit enabling conditions and probability distributions emitting from $\loc_2$ as they are irrelevant (e.g., they can encode a self-loop at $\loc_2$).
The invariant conditions $x\le 2$ and $x\le 3$ are introduced in order to prevent schedulers from repeatedly choosing time elapse. \qed
\end{example}
\vspace{-0.8em}

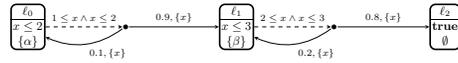
\begin{figure}
\vspace{-0.5em}
\centering
\resizebox{.5\textwidth}{!}{
\begin{tikzpicture}[x = 1.7cm]
\node[location] (task1)         at (0,0)
{
\begin{tabular}{c}
$\loc_0$\\
\hline
$x\le 2$ \\
$\{\alpha\}$
\end{tabular}
};

\node[dec] (dec1)                    at (1.5,0)  {$\bullet$};

\node[location] (task2)         at (3.2,0)
{
\begin{tabular}{c}
$\loc_1$\\
\hline
$x\le 3$ \\
$\{\beta\}$
\end{tabular}
};

\node[dec] (dec2)                    at (4.7,0)  {$\bullet$};

\node[location] (finished)         at (6.4,0)
{
\begin{tabular}{c}
$\loc_2$\\
\hline
$\true$ \\
$\emptyset$
\end{tabular}
};

\draw[tran,dashed]    (task1)    to node[auto, font=\scriptsize] {$1\le x\wedge x\le 2$}    (dec1);
\draw[tran,bend left] (dec1)     to node[auto, font=\scriptsize] {$0.1,\{x\}$}     (task1);
\draw[tran]           (dec1)     to node[auto, font=\scriptsize] {$0.9, \{x\}$}    (task2);
\draw[tran,dashed]    (task2)    to node[auto, font=\scriptsize] {$2\le x\wedge x\le 3$}     (dec2);
\draw[tran,bend left] (dec2)     to node[auto, font=\scriptsize] {$0.2,\{x\}$}     (task2);
\draw[tran]           (dec2)     to node[auto, font=\scriptsize] {$0.8, \{x\}$}    (finished);
\end{tikzpicture}
}
\caption{A Task-Completion Example}
\label{fig:taskcompletion}
\vspace{-1em}
\end{figure}

A simple specification for {\sc Task-Completion} problem is that all tasks should be finished within a given amount of time with probability at least some given number.
We consider DTA-specifications which can express also the maximal completion time over individual tasks.
Example~\ref{ex:taskcompletiondta} explains this idea.

\begin{example}\label{ex:taskcompletiondta}
Consider the DTA depicted in Figure~\ref{fig:dtataskcompletion} which works as a specification for the PTA in Example~\ref{ex:taskcompletion}.
$q_i$'s ($1\le i\le 4$) are modes with $q_3$ being the final mode, $y,z$ are clocks and arrows between modes are rules.
For example, there are two rules emitting from $q_1$, one is $(q_1, \{\beta\}, y\le 3, \{y\}, q_2)$ and the other is
$(q_1, \{\alpha\}, \true,\emptyset, q_1)$.
$q_0$ is the initial mode to read the label of the initial location of a PTA in the product construction, and
$q_3$ is the final mode.
Note that this DTA does not satisfy the totality condition. However, this can be remedied by adding rules leading to a deadlock mode without changing the acceptance behaviour of the DTA.
In the product construction with the PTA in Example~\ref{ex:taskcompletion}, $y$ records completion time of individual tasks and $z$ records completion time of both tasks.
The specification then says that the PTA should complete the first task in $3$ time units (by $y\le 3$), the second task in $4$ time units (by $y\le 4$), and all the tasks in $6$ time units (by $z\le 6$).\qed
\end{example}

\begin{figure}
\centering
\resizebox{.5\textwidth}{!}{
\begin{tikzpicture}[x = 3.8cm]
\node[mode,initial] (q0)         at (0.5,0)  {$\dtloc_0$};
\node[mode] (q1)         at (1.3,0)  {$\dtloc_1$};
\node[mode] (q2)         at (2,0)  {$\dtloc_2$};
\node[mode,accepting] (q3)         at (3,0)  {$\dtloc_3$};

\draw[tran] (q0)       to node[auto, font=\scriptsize] {$\{\alpha\}, \true, \{y\}$}      (q1);
\draw[tran, loop above] (q1)       to node[auto, font=\scriptsize] {$\{\alpha\}, \true, \emptyset$}  (q1);
\draw[tran] (q1)       to node[auto, font=\scriptsize] {$\{\beta\}, y\le 3,  \{y\}$} (q2);
\draw[tran, loop above] (q2)       to node[auto, font=\scriptsize] {$\{\beta\}, \true, \emptyset$}  (q2);
\draw[tran] (q2)       to node[auto, font=\scriptsize] {$\emptyset, y\le 4\wedge z\le 6,  \emptyset$} (q3);

\end{tikzpicture}
}
\caption{A DTA Specification for Example~\ref{ex:taskcompletion}}
\label{fig:dtataskcompletion}
\end{figure}
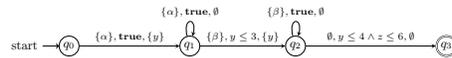
\subsection{Robot Navigation}
This case study is motivated from~\cite{DBLP:conf/tacas/BarbotCHKM11}.
In this case study, a robot is given the task to reach a destination in an unknown area.
Since the area is unknown,
the strategy the robot takes is to traverse the area randomly until the destination is reached.
Example~\ref{ex:robotnavigation} illustrates a simple setting on a $3$-by-$2$ grid.

\begin{example}\label{ex:robotnavigation}
Consider the robot-navigation problem depicted in Figure~\ref{fig:robotnavigation}.
On each tile, the time taken by a robot to leave the tile is always between $2$ and $3$ time-units.
The tile filled with black is an obstacle which cannot be entered.
The task for a robot is to start from the left-down corner of the $3$-by-$2$ grid, and move uniform-randomly to adjacent tiles excluding the obstacle until the upright corner is reached.
We assume that the robot does not always succeed to leave a tile, and the probability to successfully leave a tile is $0.9$.
The PTA modelling this problem is depicted in Figure~\ref{fig:ptarobotnavigation},
for which $x$ is the clock to measure the dwell-time on each tile, $\alpha,\beta$ are atomic propositions that distinguish adjacent tiles, and the way that the PTA is depicted is the same as for Example~\ref{ex:taskcompletion} and Figure~\ref{fig:taskcompletion}.
Each location $\loc_{i,j}$ corresponds to the situation that the robot stands in the tile $(i,j)$ (viewed as a coordinate in a two-dimensional plane) of the original grid.
The location $\loc_{2,1}$ is labelled $\emptyset$ to signify the destination.
Same as in Example~\ref{ex:taskcompletion}, we elaborate invariant conditions to disallow schedulers from repeatedly choosing time elapses. \qed
\end{example}

\begin{figure}
\begin{minipage}{0.3\textwidth}
\centering
~~\\
~\\
~\\
\scalebox{0.8}[0.8]{
\begin{tikzpicture}[x = 1cm]
\begin{scope}
    \draw (0, 1) grid (3, 3);

    \setcounter{row}{1}
    \setrow{}{}{}
    \setrow{}{}{}
\end{scope}

\fill[black] (1,2) rectangle (2,3);

\end{tikzpicture}
}
\caption{A Robot Navigation}
\label{fig:robotnavigation}
\end{minipage}
\begin{minipage}{0.6\textwidth}
\centering
\scalebox{0.5}[0.5]{
\begin{tikzpicture}[x=2cm, y=2cm]

\node[location] (q00)         at (0,0)
{
\begin{tabular}{c}
$\loc_{0,0}$\\
\hline
$x\le 3$ \\
$\{\beta\}$
\end{tabular}
};

\node           (dec00)       at (1, 1) {$\bullet$};

\node[location] (q01)         at (0,2)
{
\begin{tabular}{c}
$\loc_{0,1}$\\
\hline
$x\le 3$ \\
$\{\alpha\}$
\end{tabular}
};

\node           (dec01)       at (0, 1) {$\bullet$};

\node[location] (q10)         at (2,0)
{
\begin{tabular}{c}
$\loc_{1,0}$\\
\hline
$x\le 3$ \\
$\{\alpha\}$
\end{tabular}
};

\node (dec10) at (2,-1) {$\bullet$};

\node[location] (q20)         at (4,0)
{
\begin{tabular}{c}
$\loc_{2,0}$\\
\hline
$x\le 3$ \\
$\{\beta\}$
\end{tabular}
};

\node (dec20)  at (3,1) {$\bullet$};

\node[location] (q21)         at (4,2)
{
\begin{tabular}{c}
$\loc_{2,1}$\\
\hline
$\true$ \\
$\emptyset$
\end{tabular}
};

\node (sp1) at (0.9,0.4)  {{\scriptsize $2\le x\wedge x\le 3$}};
\node (sp2) at (3.1,0.4)  {{\scriptsize $2\le x\wedge x\le 3$}};
\node (dec00q00) at (0.4,0.9)  {{\scriptsize $0.1, \{x\}$}};
\node (dec01q01) at (0.35, 1.2)  {{\scriptsize $0.1, \{x\}$}};
\node (dec00q10) at (1.6, 0.8) {{\scriptsize $0.45, \{x\}$}};


\draw[tran]               (dec00)       to node[above right, font=\scriptsize]  {$0.45, \{x\}$}     (q01);
\draw[tran]               (dec00)       to      (q10);
\draw[tran,bend right=20] (dec00)       to      (q00);

\draw[tran]               (dec01)       to node[left, font=\scriptsize]  {$0.9, \{x\}$}               (q00);
\draw[tran,bend right=20] (dec01)       to                (q01);

\draw[tran,bend left=20]   (dec10)       to node[below left, font=\scriptsize]  {$0.45, \{x\}$}             (q00);
\draw[tran,bend right=20]  (dec10)       to node[below right, font=\scriptsize]  {$0.45, \{x\}$}             (q20);
\draw[tran,bend left=20]   (dec10)       to node[left, font=\scriptsize]  {$0.1, \{x\}$}           (q10);

\draw[tran]  (dec20)       to node[above left, font=\scriptsize]  {$0.45, \{x\}$}     (q10);
\draw[tran]  (dec20)       to node[auto, font=\scriptsize]  {$0.45, \{x\}$}     (q21);
\draw[tran, bend left=20]  (dec20)       to node[auto, font=\scriptsize]  {$0.1, \{x\}$}     (q20);

\draw[tran, dashed]  (q00)  to   (dec00);
\draw[tran, dashed]  (q01)  to node[left, font=\scriptsize]  {$2\le x\wedge x\le 3$}  (dec01);
\draw[tran, dashed]  (q10)  to node[auto, font=\scriptsize]  {$2\le x\wedge x\le 3$}  (dec10);
\draw[tran, dashed]  (q20)  to   (dec20);

\end{tikzpicture}
}
\caption{The PTA for Example~\ref{ex:robotnavigation}}
\label{fig:ptarobotnavigation}
\end{minipage}
\end{figure}

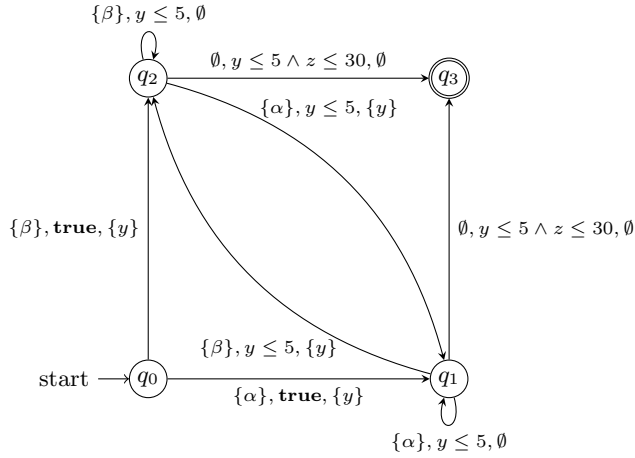
\begin{figure}
\centering
\begin{tikzpicture}[x = 4cm, y=4cm]

\node[mode,initial] (q0)         at (0,0)  {$\dtloc_0$};
\node[mode] (q1)         at (1,0)  {$\dtloc_1$};
\node[mode] (q2)         at (0,1)  {$\dtloc_2$};
\node[mode,accepting] (q3)         at (1,1)  {$\dtloc_3$};

\node (sp1) at (0.4, 0.1)  {{\scriptsize $\{\beta\}, y\le 5, \{y\}$}};
\node (sp2) at (0.6, 0.9) {{\scriptsize $\{\alpha\}, y\le 5, \{y\}$}};

\draw[tran] (q0)       to node[below, font=\scriptsize] {$\{\alpha\}, \true, \{y\}$}      (q1);
\draw[tran] (q0)       to node[left, font=\scriptsize] {$\{\beta\}, \true, \{y\}$}      (q2);
\draw[tran, loop below] (q1)      to node[auto, font=\scriptsize] {$\{\alpha\}, y\le 5, \emptyset$}  (q1);
\draw[tran, bend left]  (q1)      to    (q2);
\draw[tran, loop above] (q2)      to node[auto, font=\scriptsize] {$\{\beta\}, y\le 5, \emptyset$}   (q2);
\draw[tran, bend left]  (q2)       to    (q1);
\draw[tran] (q1)       to node[right, font=\scriptsize] {$\emptyset, y\le 5\wedge z\le 30,  \emptyset$} (q3);
\draw[tran] (q2)       to node[above, font=\scriptsize] {$\emptyset, y\le 5\wedge z\le 30,  \emptyset$} (q3);

\end{tikzpicture}
\caption{A DTA Specification for Example~\ref{ex:robotnavigation}}
\label{fig:dtarobotnavigation}
\end{figure}

Similar to the previous case study, we consider the specification that stress timing constraints on both dwell-time in individual tiles and total time to reach the destination for the robot.
\vspace{-0.5em}
\begin{example}
The DTA depicted in Figure~\ref{fig:dtarobotnavigation} specifies a property for the robot navigation described in Example~\ref{ex:robotnavigation}.
The way to render this DTA is the same as for Example~\ref{ex:taskcompletiondta}.
$q_0$ is the initial mode which reads the label of the initial tile where the robot lies and $q_3$ is the final mode.
The clock $y$ measures dwell-time on an individual tile, and the clock $z$ measures the total time to destination.
The property says that the robot should (i) never dwell on an individual tile more than $5$ time units (cf. the clock constraint $y\le 5$ and atomic propositions $\alpha,\beta$ that distinguishes adjacent tiles), and (ii) reach the upright corner within 30 time units (cf. the clock constraint $z\le 30$).\qed
\end{example}
\vspace{-0.8em}
\vspace{-1.2em}
\section{Experimental Results}
\vspace{-1em}
In this section, we perform experiments over the two case studies in Section~\ref{sect:casestudies} using the PRISM model checker~\cite{DBLP:conf/cav/KwiatkowskaNP11}.
Based on Corollary~\ref{crly:opt}, we first implement the \emph{second version} of the product construction, then use the stochastic-game engine of PRISM to compute reachability probabilities.
The experiments are executed on a Intel Xeon E5-2680 v4*2 with 14 cores and 128GB memory.
We note that although our processor has 14 cores, only one core is used since PRISM runs only on one core.

The experiments for Task-Completion are carried out by investigating our approach over instances with different number $N$ of tasks, while the experiments for Robot-Navigation are carried out with different grid sizes $N\times N$ (with random obstacles).
For Task-Completion, we consider $N$ tasks where each is guaranteed to be processes in $2$ to $3$ time units, with however randomly generated probabilities for successful completion from $0.81,0.82,\dots,0.95$; the DTA in Figure~\ref{fig:dtataskcompletion} is then adapted to have $N+2$ modes where (i) all appearances of $y\le 3$, $y\le 4$ are uniformly changed to $y\le 9$ and (ii) $z\le 6$ is changed to $z\le 6\cdot N$.
For Robot-Navigation, we following the setting in Example~\ref{ex:robotnavigation} but change the size of the grid to be $N\times N$; we still use the DTA in Figure~\ref{fig:dtarobotnavigation} with the differences that we change $y\le 5$ to $y\le 9$ and $z\le 30$ to $z\le 15\cdot N$.

The experimental results are illustrated in Table~\ref{tab:expmresults}.
We compute both the minimum acceptance probability  $\inf_\sigma \pr{\dtloc,F}{\sigma}$ and the maximum one $\sup_\sigma  \pr{\dtloc,F}{\sigma}$, which are indicated by the columns ``$p_{\min}$'' and ``$p_{\max}$'', respectively.
For Task-Completion, the column ``$N$'' indicates the number of tasks, while for Robot-Navigation, this column indicates that the grid size is $N\times N$. The column ``Size'' is the number of states in the forward reachability graph generated by the stochastic-game engine. The column ``$T_{\min}$'' (resp. ``$T_{\max}$'') is the total execution time for $p_{\min}$ (resp. $p_{\max}$) of the stochastic-game engine in PRISM measured in seconds, including the time for the product construction.
The probability values for $p_{\min}$ and $p_{\max}$ are truncated up to $10^{-4}$, while the values for $T_{\min}$ and $T_{\max}$ are truncated up to $10^{-2}$.
We stop at $N=10$ for Robot-Navigation since the total running time is already nearly 4 hours.
One can observe that the probability values for Robot-Navigation decrease drastically when $N$ increases as it is more difficult for a random-walking robot to escape the grid when the grid size is larger.

\begin{table}
\vspace{-1em}
\caption{Experimental Results for Task-Completion and Robot-Navigation}
\label{tab:expmresults}
\centering
\begin{tabular}{|c|c|c|c|c|c||||c|c|c|c|c|c|}
\hline
\multicolumn{6}{|c||||}{Task-Completion} & \multicolumn{6}{|c|}{Robot-Navigation}\\
\hline
$N$ & Size & $T_{\min}$ & $p_{\min}$ & $T_{\max}$ & $p_{\max}$ & $N$ & Size & $T_{\min}$ & $p_{\min}$ & $T_{\max}$  & $p_{\max}$ \\
\hline
5& 3385 & 0.85s & 0.9960 & 1.37s & 0.9995 & 4& 1971 & 3.62s & 0.2166 & 8.96s & 0.3799  \\
\hline
7& 11652 & 3.36s & 0.9854 & 3.71s & 0.9977 & 5& 3390 & 14.43s & 0.1672 & 37.75s & 0.3146 \\
\hline
9& 29294 & 2.68s & 0.9779 & 6.44s & 0.9966 & 6& 7368 & 77.92s & 0.0724 & 260.08s & 0.1753 \\
\hline
11& 62313 & 10.51s & 0.9731 & 13.87s & 0.9957 & 7& 9873 & 271.30s & 0.0473 & 733.48s & 0.1233 \\
\hline
13& 118139 & 27.92s & 0.9667 & 34.84s & 0.9946 & 8& 18131 & 707.03s & 0.0200 & 2281.55s & 0.0607 \\
\hline
15& 203845 & 56.07s & 0.9699 & 82.23s & 0.9957 & 9& 21564 & 1648.52s & 0.0092 & 4427.92s & 0.0371 \\
\hline
17& 330462 & 95.763s & 0.9617 & 124.87s & 0.9939 & 10& 36210 & 3219.61s & 0.0076 & 10158.55s & 0.0329 \\
\hline
19& 508880 & 175.04s & 0.9541 & 232.17s & 0.9928 & - & - & - & - & - & -\\
\hline
20& 620173 & 221.06s & 0.9507 & 227.38s & 0.9920 & - & - & - & - & - & - \\
\hline
\end{tabular}
\vspace{-1em}
\end{table}

\vspace{-1.8em}
\section{Conclusion}
\vspace{-1em}

In this paper, we studied the problem of model-checking PTAs against timed-automata specifications.
We considered Rabin acceptance condition as the acceptance criterion.
We first solved the problem with deterministic-timed-automata specifications through a product construction and proved that its computational complexity is EXPTIME-complete.
Then we proved that the problem with general timed-automata specifications is undecidable through a reduction from the universality problem of timed automata.

A future direction is zone-based algorithms for Rabin acceptance condition.
Another direction is to investigate timed-automata specifications with cost or reward.
Besides, it is also interesting to explore model-checking PTAs against Metric Temporal Logic~\cite{DBLP:journals/rts/Koymans90}.

\vspace{-1.3em}
\section{Related Works}
\vspace{-0.8em}

Model-checking TAs or MDPs against omega-regular (dense-time) properties is well-studied (cf.~\cite{DBLP:books/daglib/0020348,DBLP:conf/lics/OuaknineW05,DBLP:conf/arts/Vardi99}, etc.).
PTAs extend both TAs and MDPs with either probability or timing constraints,
hence require new techniques for verification problems.
On one hand, our technique extends techniques for MDPs (e.g.~\cite{DBLP:conf/arts/Vardi99}) with timing constraints.
On the other hand, our technique is incomparable to techniques for TAs since linear-time model checking of TAs focus mostly on proving decidability of temporal logic formulas (e.g. Metric Temporal logic~\cite{DBLP:journals/rts/Koymans90,DBLP:journals/jacm/AlurFH96,DBLP:conf/lics/OuaknineW05}),
while we prove that model-checking PTAs against TA-specifications is undecidable.

Model-checking probabilistic timed models against linear dense-time properties
are mostly considered for continuous-time Markov processes (CTMPs).
First, Donatelli~\emph{et al.}~\cite{DBLP:journals/tse/DonatelliHS09} proved an expressibility result that the class of linear dense-time properties encoded by DTAs is not subsumed by branching-time properties.
They also demonstrated an efficient algorithm for verifying continuous-time Markov chains~\cite{DBLP:journals/tse/DonatelliHS09} against one-clock DTAs.
Then various results on verifying CTMPs are obtained for specifications through DTAs and general timed automata (cf. e.g.~\cite{DBLP:journals/tse/DonatelliHS09,DBLP:journals/corr/abs-1101-3694,DBLP:conf/hybrid/Fu13,DBLP:conf/hybrid/BrazdilKKKR11,DBLP:conf/tacas/BarbotCHKM11,DBLP:conf/formats/BortolussiL15}).
The fundamental difference between CTMPs and PTAs is that the former assign probability distributions to time elapses, while the latter treat time-elapses as pure nondeterminism.
As a consequence, the techniques for CTMPs cannot be applied to PTAs.

For PTAs, the only relevant result is by Sproston~\cite{DBLP:conf/qest/Sproston11} who developed an approach for verifying PTAs against deterministic discrete-time omega-regular automata by a similar product construction.
Our results extend his result in two ways.
First, our product construction has the extra ability to tackle timing constraints from the DTA.
The extension is nontrivial since it needs to resolve the integration between randomness (from the PTA) and timing constraints (from the DTA), and still ensures the EXPTIME-completeness of the problem, matching the computational complexity in the discrete-time case \cite{DBLP:conf/qest/Sproston11}.
Second, we have proved an undecidability result in the case of general nondeterministic timed automata, thus extending \cite{DBLP:conf/qest/Sproston11} with nondeterminism.

\vspace{-1.2em}
\section{Acknowledge}
\vspace{-1em}

Acknowledgement
This work has been supported by by the National Natural Science Foundation of
China (Grants 61761136011, 61532019, 61472473).

Thanks reviewers for detailed review comments.
\vspace{-1.2em}
\bibliographystyle{splncs}
\bibliography{pta-dta}

\clearpage
\appendix

\section{Proofs for Technical Results}\label{app:proof}


\noindent{\bf Definition~\ref{def:accprob}.}
The probability that $\pta$ \emph{observes} $\tra$ under scheduler $\sigma$ and initial mode $\dtloc\in\cstates$, denoted by $\pr{\dtloc}{\sigma}$, is defined by:
\[
    \pr{\dtloc}{\sigma}
        :=
            \probm^{\pta,\sigma}(
                \LangCsAqF
            )
\]
where $\LangCsAqF$ is the set of infinite paths under $\pta$ that are accepted by the TRA $\tra$ w.r.t $(\dtloc,\zero)$ i.e.
\begin{small}
\begin{align*}
    &
    \LangCsAqF = 
        \left \{
            \infpath \in \infpaths{\pta,\sigma} \mid
            \acccept
                {\tra}
                {(\dtloc,\zero)}
                {\infpath}
        \right\}
        \\
        \bigcup_{i=1}^{n} (
            &
            \bigcup_{m \in \Nset}
            \bigcap \left \{
                \overline{\cyl(\fnpath)}
                \mid
                \fnpath \in \fnpaths{\pta,\sigma}
                \mbox{,}
                m \le \length{\fnpath}
                \mbox{,}
                \lastloc{
                    \traj{
                        \run{\tra}{\iconfig}{\lbfunc(\fnpath)}
                    }
                } \in H_i
            \right \}
            \\
            \cap
            &
            \bigcap_{m \in \Nset}
            \bigcup \left \{
                \cyl(\fnpath)
                \mid
                \fnpath \in \fnpaths{\pta,\sigma}
                \mbox{,}
                m \le \length{\fnpath}
                \mbox{,}
                \lastloc{
                    \traj{
                        \run{\tra}{\iconfig}{\lbfunc(\fnpath)}
                    }
                } \in K_i
            \right \}
        )
\end{align*}
\end{small}
where
$
    \dtloc^*
        =
            \trfunc \left(
                (\dtloc,\zero),
                \lbfunc (
                    \initloc{\fnpath}
            \right)
$.

Since the set $\fnpaths{\pta,\sigma}$ is countably-infinite,
$\LangCsAqF$ is measurable since it can be represented as a countable intersection of certain countable unions of some cylinder sets (cf.~\cite[Remark 10.24]{DBLP:books/daglib/0020348} for details).

\smallskip
\noindent{\bf Lemma~\ref{lemm:pfuncbij}.}
The function $\pfunc$ is a bijection. Moreover, for any infinite path $\infpath$ under $\pta$, $\infpath$ is non-zeno iff $\pfunc(\infpath)$ is non-zeno.
\begin{proof}
The first claim follows directly from the determinism and totality of DTAs.
The second claim follows from the fact that $\pfunc$ preserves time elapses in the transformation.
\end{proof}

\noindent{\bf Proposition~\ref{prop:psfunc}.}
For any scheduler $\sigma$ for $\pta$ and any initial mode $q$ for $\dta$, we have $\TLang = \TAcc.$
\begin{proof}
By definition, the set $\LangCsAqF$ equals
$$
    \left \{
        \infpath \in \infpaths{\pta,\sigma} \mid
        \accept{
            \infset{
                \traj{
                    \xi_\pi
                }
            }
        }   {
            \rabin
        }
    \right\}
$$
where $\xi_\pi$ is the unique run of $\dta$ on $\lbfunc(\pi)$ with initial configuration $(\dtloc^*,\zero)$ for which
$\dtloc^*$ is the unique location such that $\dtatr{(q,\zero)}{\lbfunc(\loc^*)}{(\dtloc^*,\zero)}$.
Let $\infpath = (\loc_0,\nu_0) a_0 (\loc_1,\nu_1) a_1 \dots $ be any infinite path under $\pta$.
By the definition of $\pfunc$ we have
$$
    \pfunc( \infpath ) =
        ((\loc_0,\dtloc_0),\nu_0\cup\mu_0)
        a'_0
        ((\loc_1,\dtloc_1),\nu_1\cup\mu_1)
        a'_1
        \dots
$$
in the form (\ref{eq:trinfpath}) such that $\xi_\pi=\{(\dtloc_n,\mu_n,b_{n})\}_{n\in\Nset_0}$ is the unique run
on $\lbfunc(\infpath)=b_0b_1\dots$ . Moreover, $\pi$ follows $\sigma$ iff $\pfunc( \infpath )$ follows $\sfunc(\sigma)$ by definition.
Then it is obvious that
$$\trace{ \pfunc( \infpath ) }
    = q_0 q_1 \dots
    = \traj{
        \xi_\infpath
    }.
$$
It follows that
$\infset {
    \trace {
        \pfunc( \infpath )
    }
}$
is Rabin accepting by $\rabin$ iff
$
    \infset{
        \traj {
            \xi_\infpath
        }
    }
$
is Rabin accepting by $\rabin$. Hence the result follows from Lemma~\ref{lemm:pfuncbij}.
\end{proof}
\smallskip
\noindent{\bf Proposition~\ref{thm:main}.}
For any scheduler $\sigma$ for $\pta$ and mode $q$, the followings hold:
\begin{compactitem}
\item
{\small $
    \pr
        {\dtloc}
        {\sigma}
        =
            \probm
                ^{\pta,\sigma}
                \left(
                    \LangCsAqF
                \right)
        =
            \probm
                ^{\product{\pta}{\dta_\dtloc},\theta(\sigma)}
                \left(
                    \TAcc
                \right)
    ;
$}
\item
{\small$
    \probm
        ^{\pta,\sigma}
        \left( \{
                \infpath \mid \infpath \mbox{ is zeno}
            \}
        \right)
    =
    \probm
        ^{\product{\pta}{\dta_\dtloc},\theta(\sigma)}
        \left( \{
                \infpath' \mid \infpath' \mbox{ is zeno}
            \}
        \right).
$}
\end{compactitem}
\begin{proof}
Define the probability measure $\probm'$ by: $\probm'(A)=\probm^{\product{\pta}{\dta_\dtloc},\sfunc\left(\sigma\right)}(\pfunc(A))$ for $A\in\mathcal{F}^{\pta,\sigma}$. We show that $\probm'=\probm^{\pta,\sigma}$. By \cite[Theorem 3.3]{PBMeasure}, it suffices to consider cylinder sets as they form a pi-system (cf. \cite[Page 43]{PBMeasure}).
Let $\fnpath=(\loc_0,\nu_0)a_0\dots a_{n-1}(\loc_n,\nu_n)$ be any finite path under $\pta$.
By definition, we have that
\begin{align*}
    \probm^{\pta,\sigma}(\cyl(\fnpath))
        & =
        \probm^{\product{\pta}{\dta_\dtloc}, \sfunc\left(\sigma\right)}(\cyl(\pfunc(\fnpath)))
        \\
        & =
        \probm^{\product{\pta}{\dta_\dtloc},\sfunc\left(\sigma\right)}(\pfunc(\cyl(\fnpath)))
        \\
        & =
        \probm'(\cyl(\fnpath))\enskip.
\end{align*}
The first equality comes from the fact that the product construction preserves transition probabilities. The second equality is due to $\cyl(\pfunc(\fnpath))= \pfunc(\cyl(\fnpath))$.
The final equality follows from the definition.
Hence $\probm^{\pta,\sigma}=\probm'$.
Then the first claim follows from Proposition~\ref{prop:psfunc} and the second claim follows from Lemma~\ref{lemm:pfuncbij}.
\end{proof}

\section{The Hardness Result}\label{app:hardness}

Below we prove the hardness of the PTA-DTRA problem. It is proved in \cite{LaroussinieS07} that the reachablity-probability problem for arbitrary PTAs is \emph{EXPTIME}-complete.
We show a polynomial-time reduction from the PTA reachibility problem to the PTA-DTRA problem as follows.
For an arbitrary PTA $\pta$ in the form (\ref{eq:pta})
and a set $\fstates\subseteq\locs$ of final locations,
let $\pta'=\left(\locs, \loc^*, \clocks, \acts, \inv, \enab,  \prob, \ap',\lbfunc'\right)$ where $\ap':=\ap\cup\{\mbox{\sl acc}\}$ and
$\lbfunc'$ is defined by
\begin{displaymath}
    \mathcal{L}'(l):=\left\{
    \begin{array}{cc}
        \mathcal{L}(l) & \mbox{ if } \loc\not\in\fstates\\
        \mathcal{L}(l)\cup\{\mbox{\sl acc}\} & \loc\in\fstates
    \end{array}
    \right.
\end{displaymath}
for which $\mbox{\sl acc}$ is a fresh atomic proposition.
We also construct the DTRA $\dta'$ by
\[
\dta':=\left(\{q_0,q_1\},\alphabet,\emptyset,\rules,\{(\emptyset, \{q_1\})\}\right)
\]
where $\alphabet:=\{\lbfunc'(\loc)\mid \loc\in\locs\}$ and $\rules$ contains exactly the following rules:
\begin{compactitem}
    \item  $(q_0,U,\true, \emptyset,q_1)\in\rules$ for all $U\in\alphabet$ such that $\mbox{acc}\in U$;
    \item  $(q_0,U,\true, \emptyset,q_0)\in\rules$ for all $U\in\alphabet$ such that $\mbox{acc}\not\in U$;
    \item $(q_1,U,\true, \emptyset,q_1)\in\rules$ for all $U\in\alphabet$.
\end{compactitem}
It is then straightforward from definition that an infinite path under $\pta$ visits some location in $\fstates$ iff the infinite path (under $\pta'$) is accepted by $\dta'$ under initial mode $\dtloc_0$.
Hence, under any scheduler (for both $\pta$ and $\pta'$), the probability to reach $\fstates$ in $\mathcal{C}$ equals the probability that
$\pta'$ observes $\dta'$ under initial mode $\dtloc_0$.
It follows that the problem to compute the maximum/minimum probability to reach $\fstates$ can be polynomially reduced to the PTA-DTRA problem.
Hence the problem PTA-DTRA is EXPTIME-hard.

\section{Proof for PTA-TRA Undecidability}\label{app:ptatraundecidability}

\noindent{\bf Theorem~\ref{thm:traundecidability}.}
Given a PTA $\pta$ and a TRA $\dta$, the problem to decide whether the minimal probability
that $\pta$ \emph{observes} $\dta$ (under a given initial mode) is equal to $1$ is undecidable.
%
\begin{proof}
Let $\dta=(\cstates,\alphabet,\dtclocks,\rules,\rabin)$ be any TRA where the alphabet $\alphabet = \{\ntaap{1}, \ntaap{2}, \cdots, \ntaap{k}\}$ and the initial mode is $\qstart$.
W.l.o.g, we consider that $\alphabet\subseteq 2^{\ap}$ for some finite set $\ap$.
This assumption is not restrictive since what $\ntaap{i}$'s concretely are is irrelevant, while the only thing that matters is that $\alphabet$ has $k$ different symbols.
We first construct the TRA $\dta' = (\cstates', \alphabet', \dtclocks, \rules',\rabin)$ where:

\begin{compactitem}
\item $\cstates'   = \cstates  \cup \{ \qinit \}$ for which $\qinit$ is a fresh mode;
\item $\alphabet'  = \alphabet \cup \{ \ntaap{0} \}$ for which $\ntaap{0}$ is a fresh symbol;
\item $\rules'     = \rules    \cup \{ \langle
            \qinit,
            \ntaap{0},
            \true,
            \dtclocks,
            \qstart
        \rangle
    \}$.
\end{compactitem}
Then we construct the PTA
\[
\pta'=\left(\locs, \loc^*, \clocks, \acts, \inv, \enab,  \prob, \ap, \lbfunc\right)
\]
where:
\begin{compactitem}
    \item $\locs      :=  \alphabet'$, $\loc^*     :=  \ntaap{0} $, $\clocks    :=  \emptyset $ and $\acts      :=  \alphabet $;
    \item $\inv(\ntaap{i})              :=  \true
                                            \text{ for }
                                            \ntaap{i} \in \locs$;
    \item $\enab(\ntaap{i},\ntaap{j})   :=  \true
                                            \text{ for }
                                            \ntaap{i} \in \locs
                                            \text{ and }
                                            \ntaap{j} \in \acts$;
    \item $\prob(\ntaap{i},\ntaap{j})$ is the Dirac distribution at $(\emptyset,\ntaap{j})$ (i.e., $\prob(\ntaap{i},\ntaap{j})(\emptyset,\ntaap{j})=1$ and $\prob(\ntaap{i},\ntaap{j})(X,b)=0$ whenever $(X,b)\ne(\emptyset,\ntaap{j})$),
                                            \text{ for }
                                            $\ntaap{i} \in \locs$
                                            \text{ and }
                                            $\ntaap{j} \in \acts$;
    \item $\lbfunc(\ntaap{i})           :=  \ntaap{i}
                                            \text{ for } \ntaap{i} \in \locs$.
\end{compactitem}
Note that we allow no clocks in the construction since clocks are irrelevant for our result.
Since we omit clocks, we also treat states (of $\pta'$) as single locations.
Below we prove that $\tra$ accepts all time-divergent timed words over $\Sigma$ with initial mode $\qstart$ iff
the minimal probability that $\pta'$ observes $\dta'$ with initial mode $\qinit$ equals $1$.

Consider any time-divergent infinite timed word $ w = t_0 b'_0 t_1 b'_1 \cdots $ over $\Sigma$ (where $t_i\in\Rset$ and $b'_i\in\Sigma$).
We define an infinite sequence $\{\fnpath_n\}_{n\in\Nset_0}$ of finite paths (of $\pta'$) inductively as follows:
\begin{compactitem}
\item $\fnpath_0:=b_0(=\loc^*)$; (Note that we treat states as locations since clocks are irrelevant.)
\item for $m\ge 0$, $\fnpath_{2m+1}:=\left\langle s_0,a_0,s_1,\dots,a_{k-1},s_{k},t_{m}, s_{k}\right\rangle$ if $\fnpath_{2m}=\left\langle s_0,a_0,s_1,\dots,a_{k-1},s_{k}\right\rangle$;
\item for $m\ge 0$, $\fnpath_{2m+2}:=\left\langle s_0,a_0,s_1,\dots,a_{k-1},s_{k},b'_{m}, b'_m\right\rangle$ if $\fnpath_{2m+1}=\left\langle s_0,a_0,s_1,\dots,a_{k-1},s_{k}\right\rangle$.
\end{compactitem}
Intuitively, the sequence $\{\fnpath_n\}_{n\in\Nset_0}$ is constructed by letting the PTA $\pta'$ read the timed word $w$ in a stepwise fashion, while adjusting the next location upon reading a symbol (as an action) from $\Sigma$.
Then one can define a scheduler $\sigma_w$ by:
\begin{compactitem}
\item $\sigma_w(\rho_{2m}):=t_m$ for $m\ge 0$;
\item $\sigma_w(\rho_{2m+1}):=b'_{m}$ for $m\ge 0$;
\item $\sigma_w(\rho)$ is arbitrarily defined if $\rho$ is not from the sequence $\{\fnpath_n\}_{n\in\Nset_0}$.
\end{compactitem}
Intuitively, $\sigma_w$ always chooses time-delays and actions from $w$.
From definition,
$
    \probm^{\pta',\sigma_w }\left(
        \left \{\infpath\mid \lbfunc(\infpath)=w
        \right \}
    \right)
    = 1
$.
Note that $\sigma_w$ is time divergent since $w$ is time divergent.
Hence
$$
    \pr
        {\qinit}
        {\sigma_w}
        =   \begin{cases}
            1 & \mbox{ if $\nta$ accepts $w$ with $(\qstart,\zero)$ },\\
            0 & \mbox{ if $\nta$ rejects $w$ with $(\qstart,\zero)$ }.
        \end{cases}
$$
where the underlying PTA (resp. TRA) is $\pta'$ (resp. $\dta'$).
Since all those $\sigma_w$'s correspond to all time divergent schedulers for $\pta'$, 
we have that
$
\inf_\sigma \probm^{\pta',\sigma}\left(
    \Lang
        {\pta',\sigma}
        {\nta',\qinit}
\right)
    = 1
$
iff
$\nta$ accepts all time-divergent timed words w.r.t. $(\qstart,\zero)$.
\end{proof}


\end{document}